\documentclass[letterpaper,onecolumn,10pt,accepted=2024-06-12]{quantumarticle}
\pdfoutput=1
\usepackage[square,numbers]{natbib}
\usepackage{odonnell}
\usepackage{physics,epigraph}

\hypersetup{colorlinks, 
    linkcolor={darkgray}, 
    citecolor={darkgray}, 
    urlcolor={darkgray},
	pdftitle={Quantum chi-squared tomography and mutual information testing},
	pdfauthor={Steven T. Flammia, Ryan O'Donnell}
}
\DeclarePairedDelimiterX\diverg[2]{(}{)}{#1 \mathrel{}\mathclose{}\delimsize\|\mathopen{}\mathrel{} #2}
\newcommand{\dtv}[2]{\mathrm{d}_{\mathrm{TV}}(#1,#2)}

\newcommand{\dhell}[2]{\mathrm{d}_{\mathrm{H}}(#1,#2)}
\newcommand{\dhellsq}[2]{\mathrm{d}^2_{\mathrm{H}}(#1,#2)}
\newcommand{\dren}[3]{\mathrm{d}^{\textnormal{R\'en}}_{\mathrm{#1}}\diverg{#2}{#3}}
\newcommand{\BC}[2]{\mathrm{BC}(#1,#2)}
\newcommand{\dKL}[2]{\mathrm{d}_{\mathrm{KL}}\diverg{#1}{#2}}
\newcommand{\df}[2]{\mathrm{d}_{f}\diverg{#1}{#2}}
\newcommand{\dchisq}[2]{\mathrm{d}_{\chi^2}\diverg{#1}{#2}}

\newcommand{\Dtr}[2]{\mathrm{D}_{\mathrm{tr}}(#1,#2)}

\newcommand{\DB}[2]{\mathrm{D}_{\mathrm{B}}(#1,#2)}
\newcommand{\DHell}[2]{\mathrm{D}_{\mathrm{H}}(#1,#2)}
\newcommand{\DHella}[2]{\mathrm{A}(#1,#2)}
\newcommand{\DHellsq}[2]{\mathrm{D}^2_{\mathrm{H}}(#1,#2)}
\newcommand{\DBsq}[2]{\mathrm{D}^2_{\mathrm{B}}(#1,#2)}
\newcommand{\Fid}[2]{\mathrm{F}(#1,#2)}

\newcommand{\Df}[2]{\mathrm{D}_{f}\diverg{#1}{#2}}
\newcommand{\DKL}[2]{\mathrm{S}\diverg{#1}{#2}}
\newcommand{\DBchi}[2]{\mathrm{D}_{\chi^2}\diverg{#1}{#2}}
\newcommand{\DBchihat}[2]{\widehat{\mathrm{D}}_{\chi^2}\diverg{#1}{#2}}
\newcommand{\DBchiminus}[3]{\mathrm{\widehat{D}}_{\chi^2}^{- #1}\diverg{#2}{#3}}
\newcommand{\DBon}[2]{\mathrm{D}^{\textnormal{on}}_{\chi^2}\diverg{#1}{#2}}
\newcommand{\DBoff}[2]{\mathrm{D}^{\textnormal{off}}_{\chi^2}\diverg{#1}{#2}}

\newcommand{\Dren}[3]{\mathrm{D}^{\textnormal{R\'en}}_{\mathrm{#1}}\diverg{#2}{#3}}
\newcommand{\Fro}{{\mathrm{F}}}

\newcommand{\rnote}[1]{}
\newcommand{\stf}[1]{}

\title{Quantum chi-squared tomography and mutual information testing}
\author{Steven T. Flammia}
\affiliation{AWS Center for Quantum Computing}
\affiliation{IQIM, California Institute of Technology} 
\email{stf@vt.edu}
\thanks{Present address: Department of Computer Science, Virginia Tech}
\author{Ryan O'Donnell}
\affiliation{Computer Science Department, Carnegie Mellon University} 
\email{odonnell@cs.cmu.edu}

\date{June 12, 2024}

\begin{document}

\maketitle

\begin{abstract}
For quantum state tomography on rank-$r$ dimension-$d$ states, we show that $\wt{O}(r^{.5}d^{1.5}/\eps) \leq \wt{O}(d^2/\eps)$ copies suffice for accuracy~$\eps$ with respect to (Bures) $\chi^2$-divergence, and $\wt{O}(rd/\eps)$ copies suffice for accuracy~$\eps$ with respect to quantum relative entropy. The best previous bound was $\wt{O}(rd/\eps) \leq \wt{O}(d^2/\eps)$ with respect to infidelity; our results are an improvement since
infidelity is bounded above by both the relative entropy and the $\chi^2$-divergence. 
For algorithms that are required to use single-copy measurements, we show that $\wt{O}(r^{1.5} d^{1.5}/\eps) \leq \wt{O}(d^3/\eps)$ copies suffice for $\chi^2$-divergence, and $\wt{O}(r^{2} d/\eps)$ suffice for relative entropy.\footnote{Independent and contemporaneous work~\cite{CHLLS23} achieved the bound $\wt{O}(r^2d/\eps)$ for infidelity. Indeed, the authors of that work were the first to show that $\wt{O}(d^3/\eps)$ is possible for infidelity with single-copy measurements.}
 
Using this tomography algorithm, we show that $\wt{O}(d^{2.5}/\eps)$ copies of a $d\times d$-dimensional bipartite state suffice to test if it has quantum mutual information~$0$ or at least~$\eps$. As a corollary, we also improve the best known sample complexity for the \emph{classical} version of mutual information testing to $\wt{O}(d/\eps)$.
\end{abstract}

\thispagestyle{empty}
\clearpage

\epigraph{\emph{Quantum state tomography['s] perfection is of great importance to quantum computation and quantum information.}\medskip

--- Nielsen and Chuang~\cite[p.47]{NC10}}

\section{Introduction}

Quantum state tomography --- learning a  $d$-dimensional quantum state from~$n$ copies --- is a ubiquitous task in quantum information science. 
It is the quantum analogue of the classical task of learning a $d$-outcome probability distribution from~$n$ samples.

In more detail, the goal is to design an algorithm that, given $\rho^{\otimes n}$ for some (generally mixed) quantum state $\rho \in \C^{d \times d}$, outputs (the classical description of) an estimate\footnote{Throughout this paper we use \textbf{boldface} to denote random variables.} $\wh{\brho}$ that is ``$\eps$-close'' to~$\rho$ with high probability. 
The main challenge is to minimize the sample (copy) complexity~$n$ as a function of~$d$ and~$\eps$ (and sometimes other parameters, such as $r = \rank \rho$).  
We will also be concerned with the practical issue of designing algorithms that make only single-copy (as opposed to collective) measurements.

An important aspect in specifying the quantum tomography task is the meaning of ``$\eps$-close''; i.e., what the \emph{loss} function is for judging the algorithm's estimate. 
There are many natural ways for measuring the divergence of two quantum states --- even more than for two classical probability distributions --- and the precise measure chosen can make a great deal of difference both to the necessary sample complexity, as well as to the utility of the final estimate for future applications.  

The main goal of this paper is to show a new tomography algorithm that achieves the most stringent notion of accuracy, \emph{(Bures) $\chi^2$-divergence}, while having essentially the same sample complexity as previously known algorithms using \emph{infidelity} as a loss function. 
We then given an application, to the \emph{quantum mutual information testing} problem, which crucially relies on our ability to achieve efficient state tomography with respect to $\chi^2$-divergence.

\subsection{Different quantum divergences, and prior work}  \label{sec:prior}

Let us start by recalling five important notions of ``distance'' between two \emph{classical} probability distributions $p,q$ on $[d] = \{1, 2, \dots, d\}$ (see \Cref{sec:prelims} for more details):
\begin{equation}
\label{eq:classicalineq}
    \text{($\ell_2^2$-distance)} 
    \lesssim \text{(total variation distance)}^2 \lesssim \text{Hellinger-squared ($\mathrm{H}^2$)} \lesssim
    \text{KL divergence} \lesssim
    \text{$\chi^2$-divergence}.
\end{equation}
(Here the ``$\lesssim$'' ignores small constant factors.) 
The first of these, $\ell_2^2$-distance, does not have an operational interpretation, but it is by far the easiest to calculate and reason about. 
The remainder are the ``\textbf{big~four}''~\cite[p.26]{WuITStats}: \textbf{total variation (TV) distance} controls the advantage in distinguishing $p$~from~$q$ with $1$~sample; \textbf{Hellinger-squared} controls the number of samples needed to distinguish $p$~from~$q$ with high probability; \textbf{KL divergence} has several information-theoretic interpretations; and, $\boldsymbol{\chi}^{\boldsymbol{2}}\textbf{-divergence}$ plays a central role in goodness of fit (whether an unknown~$p$ is close to a known~$q$). 
We remark that the first three quantities are bounded in $[0,1]$, but KL~divergence and $\chi^2$-divergence may be unbounded.

It is extremely easy to show (see \Cref{prop:learn-L2}) that, given $n$ samples from~$p$, the empirical estimate~$\wh{\bp}$ has expected $\ell_2^2$-distance at most~$1/n$ from~$p$; hence $n = O(1/\eps)$ samples suffices for high-probability estimation with this loss function. 
Moreover, Cauchy--Schwarz immediately bounds $\mathrm{TV}^2$ by $d$~times~$\ell_2^2$, and hence $O(d/\eps)$ samples suffice when $\eps$ denotes~$\mathrm{TV}^2$ (and $\Omega(d/\eps)$ can be proven necessary). 
But in fact, $n = O(d/\eps)$ samples suffice even when $\eps$ denotes the most stringent distance, $\chi^2$-divergence. 
This also follows from a short calculation of the expected $\chi^2$-divergence of~$\wh{\bp}$ from~$p$ when $\wh{\bp}$ is the \emph{add-one} empirical estimator (see \Cref{prop:learn-chi2}).\\

The preceding five distances have natural generalizations for quantum states $\rho, \sigma \in \C^{d \times d}$. 
The analogous chain of inequalities to \cref{eq:classicalineq} is not quite true, but we have instead
\begin{equation}
    \text{(Frobenius distance)}^2 
    \lesssim \text{(trace distance)}^2 \lesssim \text{infidelity} \lesssim
    \text{quantum relative entropy},\ \text{Bures $\chi^2$-divergence}.
\end{equation}
While both quantum relative entropy and Bures $\chi^2$-divergence are bounded from below by the infidelity, neither bounds the other by a constant~\cite{TomamichelSeyfried}. 
We remark that using the ``measured relative entropy'' rather than the ``standard'' (Umegaki) quantum relative entropy \emph{does} make the full analogous chain of inequalities hold, turning the comma above into a $\lesssim$; however, the measured relative entropy is rarely used in practice. 

In the quantum case, there is a very simple empirical estimation algorithm that achieves Frobenius-squared distance~$\eps$ with $n = O(d^2/\eps)$ samples (see \Cref{sec:simple-frob}); this algorithm has the additional practical merit that copies of~$\rho$ are measured individually and nonadaptively, meaning it uses $n$ POVMs of dimension~$d$ that are fixed in advance. 
Kueng, Rauhut, and Terstiege~\cite{KUENG201788} gave another natural algorithm of this form with a refined rank-based bound:
\begin{theorem} \label{thm:krt}
    (\cite[Thm.~2]{KUENG201788}.) There is a state tomography algorithm using nonadaptive single-copy measurements achieving expected Frobenius-squared error $O(rd/n)$ on $d$-dimensional states of rank at most~$r$.
    Hence $n = O(rd/\eps)$ samples suffice to get\footnote{With probability at least $.99$, say, by Markov's inequality.} Frobenius-squared accuracy~$\eps$.
\end{theorem}
Again, Cauchy--Schwarz implies that trace distance-squared is bounded by~$r$ times Frobenius-squared, so one immediately concludes that $n = O(r^2 d/\eps)$ copies suffice for a nonadaptive single-copy measurement algorithm achieving trace distance-squared~$\eps$.

Allowing for \emph{adaptive} single-copy measurement algorithms (in which the POVM used on the $t$th copy of~$\rho$ may be chosen based on the outcomes of the first $t-1$ measurements), it is known that for $d = 2$ (a single qubit), $n = O(1/\eps)$ measurements with one ``round'' of adaptivity suffice for estimation with infidelity~$\eps$. 
The idea for this dates back to at least~\cite{Rehacek2004}, with a proof appearing in, e.g.,~\cite[Eq.\ 4.17]{Bagan2006}. 
The case of higher~$d$ is discussed in~\cite{Pereira2018}, but no complete mathematical analysis seems to appear in the literature.  
\begin{remark}  \label{rem:however}
    However, prior to completing our work, we were informed by the authors of~\cite{CHLLS22} that they could achieve infidelity~$\eps$ with~$\wt{O}(d^3/\eps)$ single-copy measurements and logarithmically many rounds of adaptivity.
\end{remark}

Moving to quantum tomography algorithms that allow for a general collective measurement on all~$n$ copies, it would seem that some amount of representation theory is needed to get optimal results (intuitively, because $\rho^{\otimes n}$ lies in the symmetric subspace). 
The following two results were shown independently and contemporaneously:
\begin{theorem} \label{thm:ow}
    (\cite[Cor.~1.4]{odonnell2016efficient}.) 
    There is state tomography algorithm using collective measurements achieving expected Frobenius-squared error $O(d/n)$ on $d$-dimensional states.
    Hence $n = O(d/\eps)$ samples suffice to get Frobenius-squared accuracy~$\eps$.
    As a corollary of Cauchy--Schwarz, $n = O(rd/\eps)$ samples suffice to get trace distance-squared accuracy~$\eps$.
\end{theorem}
\begin{theorem} \label{thm:hhjwy}
    (\cite[(14)]{haah2017sample}.)  There is a state tomography algorithm using collective measurements on $n = O(rd/\eps) \cdot \log(d/\eps)$ copies that achieves infidelity~$\eps$.
\end{theorem}
\begin{remark}
    Except for the $\log(d/\eps)$ factor, \Cref{thm:hhjwy} is stronger than the corollary in~\Cref{thm:ow}, since $\text{(trace distance)}^2\lesssim \text{infidelity}$.
    If one wishes to have optimal $O(1/\eps)$ dependence on~$\eps$ (no log factor), the best known result is $n = O(r^2d/\eps)$ using very sophisticated representation theory~\cite{odonnell2017efficient}.  On the other hand, if one wishes to have optimal $O(rd)$ dependence (no log factor), prior to the present work the best result was $O(rd/\eps^2)$, following from \Cref{thm:ow} and $\text{infidelity} \lesssim \text{(trace distance)}^2$.
\end{remark}

Turning to lower bounds, Haah--Harrow--Ji--Wu--Yu~\cite{haah2017sample} showed that for collective measurements, $\Omega(d^2/\eps)$ samples are necessary for trace distance-squared tomography in the full-rank case, and $\Omega(\frac{rd}{\eps  \log(d/r\eps)})$ are necessary in the general rank-$r$ case; Yuen~\cite{Yue23} recently removed the log 
factor in case $\eps$ stands for infidelity.
As for single-copy measurement algorithms, \cite{haah2017sample} showed (improving on~\cite{FB15}) that for \emph{nonadaptive} algorithms, $\Omega(\frac{r^2 d}{\eps^2 \log(1/\eps)})$ copies are needed for infidelity-tomography, and $\Omega(d^3/\eps)$ copies are needed for trace distance-squared tomography in the full-rank case. This latter bound was also very recently established~\cite{CHLLS22} even in the \emph{adaptive} single-copy case.

\subsection{Our results} \label{sec:results}

A major question left open by the preceding results is whether quantum state tomography with $\wt{O}(1/\eps)$ dependence is possible for a notion of accuracy more stringent than that of infidelity, such as quantum relative entropy or $\chi^2$-divergence.
Although efficient learning with respect to these more stringent measures is known to be possible in the classical case, we are not aware of any previous provable results along these lines in the quantum case. 
Indeed, these divergences seem fundamentally more difficult to handle, not being bounded in~$[0,1]$, and prior works seemed to suggest that negative results might hold for them. 

Prior authors have considered tomography with respect to these stronger error notions. 
For example, Ferrie and Blume-Kohout~\cite{FB15} investigated qubit tomography with respect to quantum relative entropy, and Ref.~\cite{PGR23} uses $\chi^2$ hypothesis testing to study tomography of (Choi states of) quantum channels. 
A further motivation comes from the work of Blume-Kohout and Hayden~\cite{blumekohout2006accurate}, who showed that the quantum relative entropy is singled out as the unique loss function for quantum tomography once certain plausible and general desiderata of an estimator are specified. 

Our main motivation, which we return to in \Cref{sec:testing}, is a property test for zero quantum mutual information. 
For this application, our argument \emph{requires} us to do quantum state tomography with respect to Bures $\chi^2$-divergence, as only then can we use the quantum ``$\chi^2$-vs.-$\mathrm{H}^2$ identity tester'' from Ref.~\cite{Badescu2019}. 

For these two stronger error notions, we essentially show that \emph{the strongest upper bounds that one could possibly hope for indeed hold}. 
Our main theorem is the following:
\begin{theorem} \label{thm:main} 
    Suppose there exists a tomography algorithm~$\calA$ that obtains expected Frobenius-squared error at most~$f(d,r)/n$ when given $n$ copies of a quantum state $\rho \in \C^{d \times d}$ of rank at most~$r$.  
    Then it may be transformed into a tomography algorithm~$\calA'$ that, given $\eps$, $r$, and 
    \begin{equation}
        n = \wt{O}\bigl(\sqrt{rd} \cdot f(d,r)/\eps\bigr) \text{ copies } \quad \text{(respectively, } n = \wt{O}(r \cdot f(d,r)/\eps)  \text{ copies)},
    \end{equation}
    of~$\rho$, outputs (with probability at least~$.99$) the classical description of a state~$\wh{\brho}$ having 
    \begin{equation}
        \DBchi{\rho}{\wh{\brho}} \leq \eps \text{ Bures $\chi^2$-divergence accuracy} \quad \text{(respectively, }\DKL{\rho}{\wh{\brho}} \leq \eps \text{  relative entropy accuracy).}
    \end{equation}
    Moreover, if $\calA$ uses single-copy measurements, then $\calA'$ does as well, with $O(\log 1/\eps)$ rounds of adaptivity.
\end{theorem}
By plugging in \Cref{thm:ow,thm:krt}, one immediately concludes:
\begin{corollary}   \label{cor:1}
    There is a state tomography algorithm using collective measurements on $n = \wt{O}(r^{.5} d^{1.5}/\eps) \leq \wt{O}(d^2/\eps)$ copies that achieves $\chi^2$-divergence accuracy~$\eps$.
\end{corollary}
\begin{corollary}   \label{cor:1.5}
    There is a state tomography algorithm using collective measurements on $n = \wt{O}(r d/\eps)$ copies that achieves relative entropy accuracy~$\eps$.
\end{corollary}
\begin{corollary}   \label{cor:2}
    There is a state tomography algorithm using single-copy measurements and $O(\log 1/\eps)$ rounds of adaptivity on $n = \wt{O}(r^{1.5} d^{1.5}/\eps) \leq \wt{O}(d^3/\eps)$ copies that achieves $\chi^2$-divergence accuracy~$\eps$.
\end{corollary}
\begin{corollary}   \label{cor:2.5}
    There is a state tomography algorithm using single-copy measurements and $O(\log 1/\eps)$ rounds of adaptivity on $n = \wt{O}(r^2 d/\eps) \leq \wt{O}(d^3/\eps)$ copies that achieves relative entropy accuracy~$\eps$.
\end{corollary}

    Note that in the collective-measurement case, \Cref{cor:1.5} matches (up to a logarithmic factor) the $\wt{O}(rd/\eps)$ bound known previously only for infidelity-tomography, and \Cref{cor:1} also matches it in the high-rank~$r = \Theta(d)$ case.
    As for \Cref{cor:2,cor:2.5}, independent and contemporaneous work~\cite{CHLLS23} showed a weaker version of \Cref{cor:2.5} with infidelity accuracy in place of relative entropy.

\begin{remark}
    Although one would wish to achieve $\wt{O}(rd/\eps)$ scaling for $\chi^2$-tomography, we later discuss in \Cref{rem:itshard} why it seems hard to achieve dependence better than~$\wt{O}(d^{1.5}/\eps)$ even in the pure $r = 1$ case.
\end{remark}

\begin{remark}  \label{rem:111}
    In the case of~$d = 2$ (a qubit), we remove all log factors and show that $n = O(1/\eps)$ single-copy measurements with one round of adaptivity suffice for tomography with respect to $\chi^2$-divergence. 
    This simple algorithm, which illustrates the very basic idea of our \Cref{thm:main}, is given in \Cref{subsec:qubit}.
\end{remark}

\begin{remark}
    Although we have suppressed polylog factors (at most quadratic) with our $\wt{O}(\cdot)$ notation, for the case of tomography with respect to infidelity our polylog factors are actually \emph{better} than previously known in some regimes. 
    As an example, for collective measurements we have an infidelity algorithm with complexity $n = \wt{O}(\frac{rd}{\eps} \log^2(1/\eps)\log\log(1/\eps))$, which improves on the $\wt{O}(\frac{rd}{\eps} \log(d/\eps))$ bound from~\cite{haah2017sample} (and the $\wt{O}(\frac{rd}{\eps^2})$ bound following from~\cite{odonnell2016efficient})
    whenever~$\eps$ is ``large''; specifically, for $\eps \geq \exp(-\Omega(\sqrt{q}/\log q))$ in the $q$-qubit ($d = 2^q$) case. 
    See \Cref{cor:log} for details.
\end{remark}

Finally, in \Cref{sec:testing} we apply our $\chi^2$-divergence tomography algorithm to the task of \emph{testing for zero quantum mutual information}. 
In this problem, the tester gets access to~$n$ copies of a \emph{bipartite} quantum state~$\rho$ on $\C^A \otimes \C^B$ where $|A| = |B| = d$. 
The task is to accept (with probability at least~$2/3$) if the mutual information $I(A : B)_\rho$ is zero (meaning $\rho = \rho_A \otimes \rho_B$ is a product state), and to reject (with  probability at least~$2/3$) if $I(A : B)_\rho \geq \eps$. 
We show:
\begin{theorem} \label{thm:test}
    Testing for zero quantum mutual information can be done with $n = \wt{O}(1/\eps) \cdot (d^2 + r d^{1.5} + r^{.5} d^{1.75})$ samples when $\rho_A, \rho_B$ have rank at most $r \leq d$.
\end{theorem}
\begin{remark}
  The above bound is no worse than $\wt{O}(d^{2.5}/\eps)$, and is $\wt{O}(d^2/\eps)$ whenever $r \leq \sqrt{d}$. 
  One should also recall the total dimension of~$\rho$ is~$d^2$.
\end{remark}
\begin{remark}
    Harrow and Montanaro~\cite{Harrow2013} have considered a related ``product tester'' problem in the special case where the input is a pure state $|\psi\rangle$. 
    Whenever the maximum overlap $\langle\psi|\rho|\psi\rangle$ with any product state $\rho$ is $1-\eps$, the test passes with probability $1-\Theta(\eps)$ using only two copies of $|\psi\rangle$. 
    By itself however, this bound does not test quantum mutual information in the above sense, even for the rank-1 case. 
\end{remark}
\begin{remark}
    An important feature of our result is its (near-)linear scaling in~$1/\eps$. 
    This is despite the fact that estimating mutual information to~$\pm \eps$ accuracy requires $\Omega(1/\eps^2)$ samples, even for $d = 2$ and even for the classical case.
\end{remark}

Our proof of \Cref{thm:test} has two steps. 
First, we learn an estimate $\wh{\brho}_A \otimes \wh{\brho}_B$ of the marginals $\rho_A \otimes \rho_B$ that has small $\chi^2$-divergence.  
Then second we use the ``$\chi^2$-vs.-infidelity'' state certification algorithm from~\cite{Badescu2019} to test whether the unknown state~$\rho$ is close to the ``known'' state $\wh{\brho}_A \otimes \wh{\brho}_B$. 
The second step requires us to relate infidelity to relative entropy (and hence mutual information); but more crucial is that in the first step, we \emph{must} be able to do state tomography with Bures $\chi^2$-divergence as the loss measure. 
Thus we have an example where $\chi^2$-tomography is not just done for its own sake, but is necessary for a subsequent application.

Incidentally, we also show that the same two-step process works well for the problem of testing zero \emph{classical} mutual information given samples from a probability distribution~$p$ on $[d] \times [d]$:
\begin{theorem} \label{thm:test-classical}
    Testing for zero classical mutual information can be done with $n = O((d/\eps) \cdot \log(d/\eps))$ samples.  
\end{theorem}
This actually improves on the best known previous algorithm, due to Bhattacharyya--Gayen--Price--Vinodchandran~\cite{Bhattacharyya2021}, by a factor of $d \log d$.

\section{Basic results on distances and divergences} \label{sec:prelims}

\begin{notation}
    If $\rho \in \C^{d \times d}$  is a matrix and $S \subseteq [d]$, we will write $\rho[S] \in \C^{S \times S}$ for the submatrix formed by the rows and columns from~$S$.  
    If $S = [s] = \{1, 2, \dots, s\}$, we will write simply~$\rho[s]$.  We use similar notation when $p \in \R^d$ is a vector.
\end{notation}

\begin{remark}  \label{rem:triangle}
    As we frequently deal with $\ell_2^2$~error or Frobenius-squared error in this work, we often use the ``triangle inequality'' $(a-c)^2 \leq 2(a-b)^2 + 2(b-c)^2$ without additional comment.
\end{remark}

\subsection{Classical distances and divergences}

Throughout this section, let $p = (p_1, \dots, p_d)$ and $q = (q_1, \dots, q_d)$ denote probability distributions on~$[d]$. 
We also use the conventions $0/0 = 0$, $x/0 = \infty$ for $x > 0$, and trust the reader to interpret other such expressions appropriately (using continuity).

We now recall some distances between probability distributions.
\begin{definition}
    For $f : (0,\infty)\to \R$ strictly convex at~$1$ with $f(1)=0$, the associated \emph{$f$-divergence} is 
    \begin{equation}
        \df{p}{q} = \E_{\bj \sim q}[f(p_{\bj}/q_{\bj})].
    \end{equation}
\end{definition}
\begin{remark}
    All $f$-divergences satisfy the data processing inequality; see e.g.~\cite[Thm.~4.2]{WuITStats}. 
\end{remark}

\begin{definition}
    For $\alpha \in [0,\infty]$, the associated \emph{R\'enyi divergence} is defined by
    \begin{equation}
        \dren{\alpha}{p}{q} = \frac{1}{\alpha-1} \ln \sum_{i=1}^d p_i^\alpha q_i^{1-\alpha}.
    \end{equation}
\end{definition}
We will use a few particular cases:
\begin{definition}
    The \emph{total variation distance}, a metric, is the $f$-divergence with $f(x) = \tfrac12\abs{x-1}$:
    \begin{equation}
        \dtv{p}{q} = \tfrac12 \sum_{i=1}^d \abs{p_i - q_i}.
    \end{equation}
\end{definition}
\begin{definition}
    The \emph{Hellinger distance} $\dhell{p}{q}$, a metric, is the square-root of the $f$-divergence with $f(x) = (\sqrt{x}-1)^2$:
    \begin{equation}
        \dhellsq{p}{q} = \sum_{i=1}^d(\sqrt{p_i} - \sqrt{q_i})^2.
    \end{equation}
    It is also essentially a R{\'e}nyi divergence.
    More precisely, the \emph{Bhattacharyya coefficient} between $p$ and $q$ is
    \begin{equation}    \label{eqn:bhatt}
        \BC{p}{q} = \sum_{i=1}^d \sqrt{p_i} \sqrt{q_i} = \exp(-\tfrac12 \cdot \dren{1/2}{p}{q}),
    \end{equation}
    and we have $\dhellsq{p}{q} = 2(1-\BC{p}{q})$. 
    (Note the useful tensorization identity, $\BC{p_1 \otimes p_2}{q_1 \otimes q_2} = \BC{p_1}{q_1} \cdot \BC{p_2}{q_2}$.)
\end{definition}
\begin{definition}
    The \emph{KL divergence} (or \emph{relative entropy}) is both an $f$-divergence (with $f(x) = x\ln x$ or $f(x) = x \ln x - (x-1)$) and a R{\'e}nyi divergence (with $\alpha = 1$):
    \begin{equation}
        \dKL{p}{q}  = \sum_{i=1}^d p_i \ln(p_i/q_i).
    \end{equation}
    Also, if $p$ is a ``bipartite'' probability distribution on finite outcome set $A \times B$, and $p_A$, $p_B$ denote its marginals, we may define the \emph{mutual information}
    \begin{equation}
        I(A : B)_p = \dKL{p}{p_A \times p_B}.
    \end{equation}
\end{definition}
\begin{definition}  \label{def:chi2}
    The \emph{$\chi^2$-divergence} is the $f$-divergence with $f(x) = (x-1)^2$:
    \begin{equation}    \label{eqn:defchi2}
        \dchisq{p}{q} = \sum_{i=1}^d \frac{(p_i-q_i)^2}{q_i} = \parens*{\sum_{i=1}^d \frac{p_i^2}{q_i}}-1.
    \end{equation}
    We will sometimes use the first formula even when~$p$ and/or~$q$ do not sum to~$1$.
\end{definition}
\begin{definition}
    The \emph{max-relative entropy} (or \emph{worst-case regret}) is defined to be
    \begin{equation}
        \dren{\infty}{p}{q} = \max_{\substack{i \in [d] \\ p_i \neq 0}} \{\ln(p_i/q_i)\}.
    \end{equation}
\end{definition}

The following chain of inequalities is well known (see, e.g.,~\cite{gibbs2002choosing}):
\begin{proposition} \label{prop:classical-ineqs}
    $\displaystyle 
         \tfrac12\dhellsq{p}{q}
    \leq \dtv{p}{q}
    \leq \dhell{p}{q}
    \leq \sqrt{\dKL{p}{q}}
    \leq \sqrt{\dchisq{p}{q}}.
    $
\end{proposition}

Some of the inequalities in the above can be slightly sharpened; e.g., one also has  $\dtv{p}{q} \leq \sqrt{\tfrac12 \dKL{p}{q}}$, usually called \emph{Pinsker's inequality}.  
Perhaps less well known is the following ``reverse'' form of Pinsker's inequality:
\begin{equation}  \label[ineq]{ineq:rev-pinsk-class}
    \dKL{p}{q} \leq  O\bigl(\dren{\infty}{p}{q}\bigr) \cdot \dtv{p}{q}.
\end{equation}
Moreover, it is possible to strengthen the above by putting Hellinger-squared in place of total variation distance.
These facts were proven in~\cite{sason2016divergence}; for the convenience of the reader, we provide a streamlined proof of the following:
\begin{proposition} \label{prop:classical-rev-ineq}
    For $p$, $q$ probability distributions on~$[d]$ we have
    \begin{equation}
        \dKL{p}{q} \leq (2+\dren{\infty}{p}{q}) \cdot \dhellsq{p}{q}.
    \end{equation}
\end{proposition}
\begin{proof}
    Let us write $r_i = p_i/q_i$.
    Defining
    \begin{equation}
        f(r) = r \ln r - (r-1), \qquad g(r) = (\sqrt{r} - 1)^2,
    \end{equation}
    the elementary \Cref{lem:annoying-inequality} proven below  shows that
    \begin{equation} \label[ineq]{ineq:annoying}
        \forall r \geq 0, \qquad f(r) \leq h(r) g(r), \qquad \text{where } h(r) = 2+\max\{\ln r, 0\}.
    \end{equation}
    It follows that
    \begin{equation}
        \dKL{p}{q} = \E_{\bi \sim q}[f(r_{\bi})] \leq \E_{\bi \sim q}[h(r_{\bi}) g(r_{\bi})] \leq \max_{i \in [d]} \{h(r_i)\} \cdot \E_{\bi \sim q}[g(r_{\bi})] = (2 + \dren{\infty}{p}{q}) \cdot \dhellsq{p}{q}. \qedhere
    \end{equation}
\end{proof}

\begin{lemma} \label{lem:annoying-inequality} 
    \Cref{ineq:annoying} holds.
\end{lemma}
\begin{proof}
    Consider $a(r) \coloneqq h(r) g(r) - f(r)$. 
    This function is continuous and piecewise differentiable on $r \ge 0$ with an exceptional point at $r=1$. 
    We will first show that $a(r)$ is nonnegative and increasing on $r\ge 1$. 
    Clearly $a(1) = 0$, so we only need to show that $a'(r) \ge 0$. For $r\ge 1$, by the integral definition of the logarithm and the Cauchy--Schwarz inequality we have
    \begin{equation} \ln r = \int_1^r x^{-1} \mathrm{d}x \le \parens*{\int_1^r \mathrm{d}x}^{1/2} \parens*{\int_1^r x^{-2} \mathrm{d}x}^{1/2} = \sqrt{r} - \frac{1}{\sqrt{r}}\,.\end{equation}
    Calculating the derivative of $a(r)$ and using the above inequality for the logarithm, we have that 
    \begin{equation}a'(r) = 3-\frac{4}{\sqrt{r}}+\frac{1}{r}-\frac{\ln r}{\sqrt{r}} \ge 3-\frac{4}{\sqrt{r}}+\frac{1}{r}-\frac{\sqrt{r} - \frac{1}{\sqrt{r}}}{\sqrt{r}} = \frac{2 \left(\sqrt{r}-1\right)^2}{r} \ge 0.\end{equation}
    
    For the case $0 \le r \le 1$, we change variables to $s = 1/r$ and define $b(s) \coloneqq h(1/s) g(1/s) - f(1/s)$ for $s \ge 1$. 
    We have $b(1) = 0$, and using again the logarithm inequality we find
    \begin{equation}b'(s) = \frac{2 \sqrt{s}-\ln s-2}{s^2} \ge \frac{\sqrt{s}+\frac{1}{\sqrt{s}}-2}{s^2} = \frac{\left(\sqrt{s}-1\right)^2}{s^{5/2}} \ge 0. \qedhere\end{equation}
\end{proof}

We remark that \Cref{ineq:annoying} can be strengthened to $h(r) = 2 + \ln\bigl((2+r)/3\bigr)$, but as this does not change the scaling of any of our results, we will not use this stronger inequality or present our (annoyingly complicated) proof.\\

Finally, we mention the $\ell_2^2$-distance between probability distributions, $\|p - q\|_2^2 = \sum_i (p_i - q_i)^2$.  
Though it does not have an operational meaning, the simplicity of computing it makes it a useful tool when analyzing other distances. 
For example, the $\chi^2$-divergence is a kind of ``weighted'' version of $\ell_2^2$-distance, in which the error term $(p_i - q_i)^2$ is weighted by~$1/q_i$.
We record here basic facts about estimation with respect to these distance measures. 
\begin{proposition} \label{prop:learn-L2}
    Let $p$ be a distribution on $[d]$, and suppose $\bq$ is the empirical estimator formed from~$m$ samples. 
    Then for any $S \subseteq [d]$ we have
    \begin{equation}
        \E\bigl[\|p[S] - \bq[S]\|_2^2\bigr] \leq \|p[S]\|_1/m.
    \end{equation}
    In particular,  $\E\bigl[\|p - \bq\|_2^2\bigr] \leq 1/m$.
\end{proposition}
\begin{proof}
    Note that $\bq_i$ is distributed as $\textrm{Binomial}(m, p_i)/m$, hence $\E[(p_i - \bq_i)^2] = \Var[\bq_i] = p_i(1-p_i)/m \leq p_i/m$. 
    Summing over~$i \in S$ completes the proof.
\end{proof}
We will also need the following variant with high-probability guarantees:
\begin{proposition} \label{prop:learn-L2-high}
    In the setting of \Cref{prop:learn-L2}, we have the following guarantees, where we introduce the notation $m_\delta = m/(c \ln(1/\delta))$ for $c$ some large universal constant:
    \begin{enumerate}[label=(\alph*)]
        \item $\|p - \bq\|_2^2 \leq 1/m_{\delta}$ except with probability at most~$\delta$; \label{item:11}
        \item if $\|p[S]\|_1 \geq 1/m_\delta$ then $\|\bq[S]\|_1$ is within a $1.01$-factor of~$\|p[S]\|_1$ except with probability at most~$\delta$; \label{item:22}
        \item if $\|p[S]\|_1 \leq 1/m_\delta$ then $\|\bq[S]\|_1 \leq 1.01/m_{\delta}$ except with probability at most~$\delta$. 
        \label{item:33}
    \end{enumerate}
\end{proposition}
\begin{proof}
    \Cref{item:22,item:33} follow from standard Chernoff bounds.
    As for \Cref{item:11}, it follows from the known high-probability bound for empirically learning a distribution with respect to~$\ell_2^2$-error; see, e.g.,~\cite{Per19}. 
    We remark that it is important to use this latter result, as opposed to the generic ``median-of-$O(\log(1/\delta))$-estimates'' method; if we used the latter, it would be unclear how to simultaneously achieve \Cref{item:22,item:33}
\end{proof}
\begin{proposition} \label{prop:learn-chi2}
    Fix a subset $S \subseteq [d]$ of cardinality~$s$.
    Given $m$ samples from an unknown distribution $p$ on $[d]$, let $\bq$ be the estimator formed by using the add-one estimator on elements from~$S$, and the empirical estimator on the remaining elements.
    (Note that $\bq$ is itself a probability distribution.)
    Then
    \begin{equation}    \label[ineq]{ineq:ym}
        \E[\dchisq{p[S]}{\bq[S]}] \leq \frac{s}{m+s} + \parens*{\frac{(s-1)^2 }{(m+1)(m+s)} - \frac{1}{m+s}} \|p[S]\|_1 \leq s/m + (s/m)^2 \leq 2s/m,
    \end{equation}
    the last inequality assuming $m \geq s$.
    In case $S = [d]$, the sharpest upper bound above equals $\frac{d-1}{m+1} \leq d/m$.
    
    Moreover, still assuming $m \geq s$ and using the notation~$m_\delta$ from \Cref{prop:learn-L2-high}, if $p_i \geq 1/m_{\delta/s}$ for all $i \in S$, then except with probability at most~$\delta$ we have that $\bq_i$ is within a $4$-factor of~$p_i$ simultaneously for all $i \in S$.
\end{proposition}
\begin{proof}
    For $i \in S$ we have that $\bq_i$ is distributed as $\frac{\bB + 1}{m + s}$, where $\bB \sim \textrm{Binomial}(m,p_i)$.  It is elementary to show that the resulting contribution to $\dchisq{p[S]}{\bq[S]}$, namely $\frac{(p_i - \bq_i)^2}{\bq_i}$, has expectation equal to
    \begin{equation}
        \frac{1}{m+s} + \parens*{\frac{(s-1)^2}{(m+1)(m+s)} - \frac{1}{m+s} - \frac{m+s}{m+1}(1-p_i)^{m+1}}p_i.
    \end{equation}
    Dropping the term above involving $(1-p_i)^{m+1}$, and then summing over $i \in S$, yields \Cref{ineq:ym}.

    As for the ``moreover'' statement, a Chernoff bound tells us that $\bB$ is within a $2$-factor of $mp_i$ except with probability at most~$\delta/s$, using $m p_i \geq c \log(s/\delta)$.  When this occurs, $\bq_i$ is at least $p_i/4$ (using $m \geq s$) and at most $\frac{2 m p_i + 1}{m} \leq 3p_i$ (using $c \geq 1$), so the proof is complete by a union bound over $i \in S$.
\end{proof}

\subsection{Quantum distances and divergences}

The analogous theory of distances and divergences between quantum states is quite rich~\cite{tomamichel2016quantum,HM17}, as there are multiple quantum generalizations of both $f$-divergences and R{\'e}nyi divergences.
To distinguish between the quantum and classical cases, we use an upper-case~$D$ for quantum divergences and a lower-case~$d$ for classical divergences. 

Throughout this section, let $\rho, \sigma \in \C^{d \times d}$ be (mixed) quantum states.

\begin{definition}
    Given an $f$-divergence $\df{{\cdot}}{{\cdot}}$, the associated \emph{measured (aka minimal) quantum $f$-divergence}~\cite{HM17} is 
    \begin{equation}
        \Df{\rho}{\sigma} = \sup_{\text{POVMs } (E_i)_{i=1}^m} \{\df{q_\rho}{q_\sigma}\}, \quad \text{where } q_{\tau} = (\tr(\tau E_1), \dots, \tr(\tau E_m)).
    \end{equation}
\end{definition}
\begin{remark}
    All measured $f$-divergences satisfy the (quantum) data processing inequality. 
    This fact follows from the definition and a reduction to the classical case.
\end{remark}
\begin{definition}
    For $\alpha \in [0,\infty]$, the associated \emph{conventional quantum R\'enyi divergence}~\cite{Petz1986} is defined by
    \begin{equation}
        \Dren{\alpha}{\rho}{\sigma} = \frac{1}{\alpha-1} \ln\tr\parens*{\rho^\alpha \sigma^{1-\alpha}}.
    \end{equation}
\end{definition}
Let us also describe a further relationship between classical and quantum R\'enyi entropies.
To do so let us introduce the following notation:
\begin{definition}
    Given the spectral decompositions
    \begin{equation}
        \rho = \sum_{i=1}^d p_i \ketbra{\varphi_i}{\varphi_i}, \qquad \sigma = \sum_{i=1}^d q_i \ketbra{\psi_i}{\psi_i},
    \end{equation}
    we define two probability distributions $P^{\rho\sigma}, Q^{\rho\sigma}$ on $[d] \times [d]$, as follows:
    \begin{equation}
        P^{\rho\sigma}_{ij} = \abs{\braket{\varphi_i}{\psi_j}}^2 p_i, \qquad Q^{\rho\sigma}_{ij} = \abs{\braket{\varphi_i}{\psi_j}}^2 q_j.
    \end{equation}
\end{definition}
We now give a simple calculation that allows us to compute a quantum R\'{e}nyi divergence from an associated classical probability distribution. 
This calculation has appeared in the literature as early as \cite[Thm.~2.2]{nussbaum2009chernoff}; see \cite[Prop.~1]{audenaert2008asymptotic} for an explicit statement.
For convenience, we repeat the calculation here.
\begin{proposition} \label{prop:q-to-c}
    $\displaystyle \Dren{\alpha}{\rho}{\sigma} = \dren{\alpha}{P^{\rho\sigma}}{Q^{\rho\sigma}}.$
\end{proposition}
\begin{proof} One calculates:
    \begin{align*}
        \Dren{\alpha}{\rho}{\sigma} &= \frac{1}{\alpha-1} \ln \tr\parens*{\parens*{\sum_{i\vphantom{j}=1}^d p_i^\alpha \ketbra{\varphi_i}{\varphi_i}}\parens*{\sum_{j=1}^d q_j^{1-\alpha} \ketbra{\psi_j}{\psi_j}}} \\
        &= \frac{1}{\alpha-1} \ln \parens*{\sum_{i,j=1}^d p^{\alpha} \abs{\braket{\varphi_i}{\psi_j}}^2 q^{1-\alpha}}
        = \frac{1}{\alpha-1} \ln \sum_{i,j=1}^d (P^{\rho\sigma}_{ij})^\alpha (Q^{\rho\sigma}_{ij})^{1-\alpha} = \dren{\alpha}{P^{\rho\sigma}}{Q^{\rho\sigma}}. \qedhere
    \end{align*}
\end{proof}

We now define some particular quantum distances/divergences:
\begin{definition}
    The \emph{trace distance}, a metric, is the measured $f$-divergence associated to total variation distance~\cite{Helstrom1976}:
    \begin{equation}
        \Dtr{\rho}{\sigma} = \tfrac12 \|\rho - \sigma\|_1.
    \end{equation}
\end{definition}
\begin{definition}
    The \emph{Bures distance} $\DB{\rho}{\sigma}$, a metric, is the square-root of the measured $f$-divergence associated to Hellinger-squared~\cite{Fuchs1995}. 
    It has the formula
    \begin{equation}
        \DBsq{\rho}{\sigma} = 2(1-\Fid{\rho}{\sigma}),
    \end{equation}
    where $\Fid{\rho}{\sigma} = \|\sqrt{\rho}\sqrt{\sigma}\|_1$ is the \emph{fidelity} between $\rho$ and $\sigma$ (in the ``square root'' convention).  The \emph{infidelity} between $\rho$~and~$\sigma$ is simply $1-\Fid{\rho}{\sigma} = \frac12 \DBsq{\rho}{\sigma}$.
\end{definition}
There is a close analogy between the quantum fidelity and the classical Bhattacharrya coefficient, and indeed the analogue of \Cref{eqn:bhatt} holds if one uses the ``sandwiched R{\'e}nyi entropy''~\cite{mullerlennert2013quantum,wilde2014strong}.
Using instead the conventional R{\'e}nyi entropy yields a slightly different notion:
\begin{definition}  \label{def:qhell}
    The \emph{quantum Hellinger affinity} is defined by
    \begin{equation}
        \DHella{\rho}{\sigma} = \tr(\sqrt{\rho} \sqrt{\sigma}) = \exp(-\tfrac12 \cdot \Dren{1/2}{\rho}{\sigma}),
    \end{equation}
    and the \emph{quantum Hellinger distance}
    $\DHell{\rho}{\sigma}$, a metric, is defined by
    \begin{equation}
        \DHellsq{\rho}{\sigma} = 2(1-\DHella{\rho}{\sigma}) =  \tr((\sqrt{\rho} - \sqrt{\sigma})^2) = \|\sqrt{\rho} - \sqrt{\sigma}\|_{\mathrm{F}}^2 = \DHellsq{\rho}{\sigma} = \dhellsq{P^{\rho\sigma}}{Q^{\rho\sigma}},
    \end{equation}
    the last equality using \Cref{prop:q-to-c}.
    (Note also the useful tensorization identity, $\DHella{\rho_1 \otimes \rho_2}{\sigma_1 \otimes \sigma_2} = \DHella{\rho_1}{\sigma_1} \cdot \DHella{\rho_2}{\sigma_2}$.)
\end{definition}
Fortunately, the preceding two distances differ by only a small constant factor:
\begin{fact}    \label{fact:hellDB}
    $\displaystyle \DBsq{\rho}{\sigma} \leq \DHellsq{\rho}{\sigma} \leq 2 \DBsq{\rho}{\sigma}.$
\end{fact}
\noindent The left inequality in \Cref{fact:hellDB} is from $\DHella{\rho}{\sigma} \leq \Fid{\rho}{\sigma}$; the right inequality follows  from~\cite[Eq.~(32)]{audenaert2008asymptotic}.

\begin{definition}  \label{def:burchi}
    The \emph{Bures $\chi^2$-divergence} of $\rho$ from~$\sigma$ is the measured $f$-divergence associated to classical $\chi^2$-divergence~\cite{Braunstein1994,Temme2015}. 
    It can also be given the following formula when $\sigma = \diag(q_1, \dots, q_d)$ is diagonal of full rank (and this suffices to define it for general full-rank~$\sigma$, since it is unitarily invariant):
    \begin{equation}
        \label{eqn:chi-formula}
        \DBchi{\rho}{\sigma} = \sum_{i,j=1}^d \frac{2}{q_i + q_j} \abs{\tau_{ij}}^2, \quad \text{where } \tau = \rho - \sigma.
    \end{equation}
    We will use this formula even when~$q_1, \dots, q_d \geq 0$ do not sum to~$1$.
\end{definition}
Similar to the connection between $\ell_2^2$-distance and $\chi^2$-divergence in the classical case, the Bures $\chi^2$-divergence can be seen as a kind of ``weighted'' version of the Frobenius-squared distance, in which the error term $\abs{\tau_{ij}}^2$ is weighted by~$\frac{2}{q_i + q_j} = \Theta(\frac{1}{\max\{q_i,q_j\}})$.

Indeed, we will frequently consider applying \Cref{eqn:chi-formula} when the~$q_i$'s form (or approximately form) a nondecreasing sequence, meaning that (we expect) $q_i \leq q_j$.  In this case it is reasonable to use $q_i + q_j \geq q_j$, which motivates the following simple bound:
\begin{definition}
    In the notation from \Cref{def:burchi}, we define
    \begin{equation}
        \DBchihat{\rho}{\sigma} = \sum_{i,j=1}^{d} \frac{2}{q_{\max(i,j)}} \abs{\tau_{ij}}^2 \geq \DBchi{\rho}{\sigma};
    \end{equation}
    and, for $L = [d']$ (for $d' \leq d$) we define
    \begin{equation}
        \DBchiminus{L}{\rho}{\wt{\brho}} =  \sum_{\substack{i,j : \\ \max(i,j) \not \in L}} \frac{2}{q_{\max(i,j)}} \abs{\tau_{ij}}^2 = \sum_{k \not \in L} \frac{2}{q_{k}} \sum_{\substack{i,j : \\ \max(i,j) = k}} \abs{\tau_{ij}}^2.
    \end{equation}
    Note that 
    \begin{equation}    \label[ineq]{ineq:dbl}
        \DBchi{\rho}{\sigma} \leq \DBchi{\rho[L]}{\sigma[L]}  +    \DBchiminus{L}{\rho}{\sigma}.
    \end{equation}
\end{definition}

\begin{definition}
    The \emph{quantum relative entropy}~\cite{Umegaki1962}
    is defined by
    \begin{equation}
        \DKL{\rho}{\sigma} 
        = \tr\parens*{\rho(\ln \rho - \ln \sigma)}
        = \Dren{1}{\rho}{\sigma}
        = \dKL{P^{\rho\sigma}}{Q^{\rho\sigma}},
    \end{equation}
    the last equality using \Cref{prop:q-to-c}.
    Also, if $\rho$ is a ``bipartite'' quantum state on $A \otimes B$, where $A \cong B \cong \C^{d}$, and if $\rho_A$, $\rho_B$ denote its marginals (obtained by tracing out the~$B$, $A$ components, respectively), the \emph{quantum mutual information} of~$\rho$ is defined to be
    \begin{equation}
        I(A : B)_\rho = \DKL{\rho}{\rho_A \otimes \rho_B}.
    \end{equation}
\end{definition}
\begin{fact} \label{fact:Dren2}
    The conventional quantum $\infty$-R\'enyi divergence (discussed in, e.g.,~\cite{Androulakis2022}) is, by \Cref{prop:q-to-c},
    \begin{equation}
        \Dren{\infty}{\rho}{\sigma} = \max_{\substack{i,j \in [d]\\\braket{\varphi_i}{\psi_j} \neq 0}} \{\ln(p_i/q_j)\} \leq \ln(\|\rho\| \|\sigma^{-1}\|) \leq \ln \|\sigma^{-1}\|.
    \end{equation}
\end{fact}
\begin{remark}
    This quantity is \emph{not} the same as the ``quantum max-relative entropy'' defined in~\cite{Datta2009}; it \emph{would} be if one replaced the conventional R\'enyi entropy with its sandwiched form.
\end{remark}

Relating some of these divergences is the following chain of inequalities:
\begin{proposition} \label{prop:quantum-ineqs}
    $\displaystyle 
         \tfrac12\DHellsq{\rho}{\sigma}
    \leq \Dtr{\rho}{\sigma}
    \leq \DB{\rho}{\sigma}
    \leq \sqrt{\DKL{\rho}{\sigma}},\ \sqrt{\DBchi{\rho}{\sigma}}.$
\end{proposition}
\noindent The first inequality above is from  from~\cite[Thm.~2]{audenaert2008asymptotic}. The second follows from the classical case~\cite{Fuchs1995}. The third also follows from the classical case and the observation that the ``measured'' quantum relative entropy is at most $\DKL{\rho}{\sigma}$ (see, e.g.~\cite[App.~A]{BFT17}).
The fourth also follows from the classical case, using that Bures~$\chi^2$ is the measured form of classical~$\chi^2$~\cite{Braunstein1994,Temme2015}.
As with \Cref{prop:classical-ineqs}, some of these inequalities can be sharpened slightly; for example we have the quantum Pinsker inequality $\Dtr{\rho}{\sigma} \leq \sqrt{\frac{1}{2}\DKL{\rho}{\sigma}}$.

\subsection{Quantum Tomography with Quantum Relative Entropy Loss}

One of our main results follows easily from the above discussion of divergences. 
The idea is to improve on certain ``reverse quantum Pinsker'' results which have been studied previously; see, e.g., \cite{audenaert2011continuity} for a quantum generalization of the reverse-Pinsker \Cref{ineq:rev-pinsk-class}.
We will use the following strengthened version with quantum Hellinger-squared in place of trace distance:
\begin{theorem}   \label{thm:quantum-rev-ineq}
    For $\rho, \sigma \in \C^{d \times d}$, we have
    $\displaystyle 
        \DKL{\rho}{\sigma} \leq (2+ \Dren{\infty}{\rho}{\sigma}) \cdot \DHellsq{\rho}{\sigma}.$
\end{theorem}
\begin{proof}
    This is immediate from $\DKL{\rho}{\sigma} = \dKL{P^{\rho\sigma}}{Q^{\rho\sigma}}$, \Cref{prop:classical-rev-ineq}, and \Cref{fact:Dren2}.
\end{proof}
Despite following directly from known results (up to constant factors), the above theorem does not seem to have appeared previously in the literature.  
Our next result shows that this can be used to automatically upgrade any quantum tomography algorithm with an infidelity guarantee to one with a relative entropy guarantee, at the expense of only a log factor (cf.~our main \Cref{thm:main} upgrading Frobenius-squared-tomography to $\chi^2$-tomography).
\begin{notation}    \label{not:depolarizing}
    We write $\Delta_\epsilon$ for the completely depolarizing channel, which for $0 \leq \epsilon \leq 1$ acts on $\rho \in \C^{d \times d}$ as $\Delta_\epsilon(\rho) = (1-\epsilon)\rho + \epsilon (\Id/d)$ (with $\Id$ denoting the identity matrix).
\end{notation}
\begin{theorem} \label{thm:learn-KL-easy}
    Let $\calA$ be a state tomography algorithm that, given $n$~copies of $\rho \in \C^{d \times d}$ and parameter~$\eps$, outputs an estimate~$\wh{\brho}$ achieving infidelity $\frac12 \DBsq{\rho}{\wh{\brho}} \leq \eps \leq 1/2$.
    Then letting $\brho' = \Delta_{2\eps}(\wh{\brho})$, we have \mbox{$\DKL{\rho}{\brho'} \leq 16\eps \cdot (2 + \ln(d/2\eps))$.}
\end{theorem}
Applying this theorem with the previously known result of Haah--Harrow--Ji--Wu--Yu,  \Cref{thm:hhjwy}, we immediately conclude \Cref{cor:1.5}, that there is a state tomography algorithm with respect to quantum relative entropy that has copy complexity $n = O(rd/\eps) \cdot \log^2(d/\eps)$ (using collective measurements).

\Cref{thm:learn-KL-easy} is immediate from the following (together with the fact that Hellinger-squared is upper bounded by $4$~times infidelity (\Cref{fact:hellDB})):
\begin{proposition}
    Suppose $\rho, \sigma \in \C^{d \times d}$ are quantum states with $\DHellsq{\rho}{\sigma} \leq \eps$.  Then for $\sigma' = \Delta_{\eps/2}(\sigma)$ we have $\DKL{\rho}{\sigma'} \leq 4\eps \cdot (2 + \ln(2d/\eps))$.
\end{proposition}
\begin{proof}
    Since $\Delta_{\eps/2}(\sigma)$ has smallest eigenvalue at least $\eps/2d$, we have $\Dren{\infty}{\rho}{\sigma'} \leq \ln(2d/\eps)$ and hence from \Cref{thm:quantum-rev-ineq} it suffices to show $\DHellsq{\rho}{\sigma'} \leq 4\eps$.  In turn, since $\DHell{{\cdot}}{{\cdot}}$ is a metric, by \Cref{rem:triangle} it suffices to prove $\DHellsq{\sigma}{\sigma'} \leq \eps$.  
    But using \Cref{prop:quantum-ineqs}, we indeed have 
    \begin{equation}
        \DHellsq{\sigma}{\sigma'} \leq \|\sigma - \sigma'\|_1 = (\epsilon/2) \|\sigma - \Id/d\|_1 \leq (\epsilon/2)(\|\sigma\|_1 + \|\Id/d\|_1) = \epsilon. \qedhere
    \end{equation}
\end{proof}

\section{Quantum state tomography}  \label{sec:qubit}

We give a guide to this section:
\begin{itemize}
    \item  \Cref{subsec:qubit} gives a simple $\chi^2$-tomography algorithm for qubits; it achieves copy complexity $n = O(1/\eps)$ (no logs) using single-copy measurements with one round of adaptivity. 
    It serves as a small warmup for our main algorithm.
    \item \Cref{sec:frob} begins the main exposition of our reduction from Frobenius-squared-tomography to $\chi^2$-tomography. 
    This section shows how to give several useful black-box ``upgrades'' to any Frobenius-squared estimator.
    \item In \Cref{sec:sketch} we give a high-level sketch of the central estimation routine for our main theorem, which takes Frobenius-squared-tomography and turns it into a $\chi^2$-tomography algorithm ``except for very small eigenvalues''.
    \item The most involved \Cref{sec:central} follows; it fills in all the technical details for the preceding sketch.
    \item \Cref{sec:conclude} shows how to take the newly-established central estimation routine and massage its output to achieve either good $\chi^2$-accuracy (with one set of parameters) or good relative entropy accuracy (with another set of parameters).  It is in this section that we establish all the theorems and corollaries from \Cref{sec:results}.
    \item Finally, for the convenience of the reader, \Cref{sec:simple-frob} gives a simple Frobenius-squared-tomography algorithm using single-copy measurements with complexity $n = O(d^2/\eps)$.
\end{itemize}

\subsection{Qubit tomography with single-copy measurements} \label{subsec:qubit}

As mentioned in \Cref{sec:prior}, it has long been known that  one can learn a single qubit state~$\rho$ to infidelity~$\eps$ using \emph{single-copy} measurements on $O(1/\eps)$ copies of~$\rho$, combined with one ``round'' of adaptivity. 
In this section we give a short proof of the same result but with a stronger conclusion: $\eps$~accuracy with respect to Bures $\chi^2$-divergence.

We first repeat \Cref{prop:learn-chi2} in the simpler context of $d = 2$, and at the same time achieving a concentration bound:
\begin{lemma}   \label{lem:bit-chi2}
    There is a simple classical estimation algorithm that, given $n = O(\log(1/\delta)/\eps)$ samples from an unknown probability distribution $p = (p_0,p_1)$ on $\{0,1\}$, outputs an estimate $\wh{\bp}$ satisfying $\dchisq{p}{\wh{\bp}} \leq \eps$ except with probability at most~$\delta$.
\end{lemma}
\begin{proof}
    As shown in \Cref{prop:learn-chi2}, if $\wh{\bq}$ is the ``add-one estimator'' formed from~$m \geq 4/\eps$ samples, then $\E[\dchisq{p}{\wh{\bq}}] \leq \frac{1}{m+1} \leq \eps/4$. 
    By Markov's inequality, the estimator is ``good'', meaning $\dchisq{p}{\wh{\bq}} \leq \eps$, except with probability at most~$1/4$. 
    If we now use $n = O(\log(1/\delta)/\eps)$ samples to produce~$O(1/\delta)$ independent such estimators, a Chernoff bound tells us that, except with probability at most~$\delta$, at least a~$2/3$ fraction of them are ``good''. 
    If we now associate each of our estimates~$\wh{\bq} = (\wh{\bq}_0, \wh{\bq}_1)$ with the point $\wh{\bq}_1$ in the interval~$[0,1]$, we see that all of the ``good'' points appear consecutively. 
    (That is, ``reading left-to-right'', the $\wh{\bq}_1$ values consist of some ``bad'' points, followed by some ``good'' points, followed by some ``bad'' points.)
    The reason for this is that $\dchisq{p}{\wh{\bq}}$ is a monotonic function of $|p_1 - \wh{\bq}_1|$.  
    Thus if the algorithm now selects the median $\wh{\bq}_1$ point (and its associated estimate $\wh{\bq} = (\wh{\bq}_0, \wh{\bq}_1)$), this will be among the $2/3$-fraction ``good'' points except with probability at most~$\delta$.
\end{proof}

\begin{theorem} \label{thm:qubit-learning}
    There is an efficient quantum state tomography algorithm that uses $n = O(\log(1/\delta)/\eps)$ copies of an unknown qubit state $\rho \in \C^{2 \times 2}$ and outputs an estimate $\wh{\brho}$ satisfying $\DBchi{\rho}{\wh{\brho}} \leq \eps$ except with probability at most~$\delta$.  
    Moreover, the algorithm is simple to implement in the following sense:  The first $n/4$ copies of~$\rho$ are separately measured in the Pauli~$X$ basis, the next $n/4$ in the Pauli~$Y$ basis, the next $n/4$ in the Pauli~$Z$ basis, and the final~$n/4$ in a fixed basis determined by the first~$3n/4$ measurement outcomes.
\end{theorem}
\begin{proof}
    The first phase of the algorithm (using $3n/4$ copies) 
    can employ any standard single-copy quantum state tomography routine; the specific one we describe in \Cref{sec:simple-frob} has the stated Pauli format, and (using \Cref{prop:learn-L2-high}) will return a PSD matrix~$\brho'$ (not necessarily a state) satisfying
    \begin{equation}     \label[ineq]{ineq:phase1}
        \|\rho - \brho'\|_{\mathrm{F}}^2 \leq \eps/4
    \end{equation} 
    except with probability at most~$\delta/2$.
    Next, the algorithm employs a change of basis so as to make~$\brho'$ diagonal.
    It suffices to estimate~$\rho$ in this new basis. 
    Since Frobenius-distance is unitarily invariant, in the new basis \Cref{ineq:phase1} implies
    \begin{equation}    \label[ineq]{ineq:phase1a}
        |\rho_{01}|^2 = |\rho_{10}|^2 \leq \eps/8.
    \end{equation}
    As for the diagonal entries $p = (\rho_{00}, \rho_{11})$ of~$\rho$, the algorithm measures its remaining $n/4$ copies of~$\rho$ in the diagonal basis and employs \Cref{lem:bit-chi2}.
    For $n = O(\log(1/\delta)/\eps)$, this produces an estimate $\wh{\bp}$ satisfying $\dchisq{p}{\wh{\bp}} \leq \eps/2$ except with probability at most~$\delta/2$.
    The final estimate of~$\rho$ (in the new basis) will be~$\wh{\brho} = \diag(\wh{\bp})$.

    Except with probability at most~$\delta/2 + \delta/2 = \delta$, both components of the preceding algorithm produce good estimates.  Then using \Cref{eqn:chi-formula} we may decompose $\DBchi{\rho}{\wh{\brho}}$ into the on-diagonal contribution,  which is 
    $\dchisq{p}{\wh{\bp}} \leq \eps/2$, and the off-diagonal contribution, which is $2|\rho_{01}|^2 + 2|\rho_{10}|^2 \leq \eps/2$ (by \Cref{ineq:phase1a}.  
    This completes the proof.
\end{proof}

\subsection{Upgrading Frobenius-squared tomography algorithms}   \label{sec:frob}

\begin{definition}
    A function $f$ mapping quantum states to numbers at least~$1$ will be called a \emph{rate function}.  
\end{definition}
\begin{definition}
    We say a quantum state estimation algorithm~$\calA$ has \emph{Frobenius-squared rate~$f$} if the following holds:
    Whenever $\calA$ is given $m \in \N^+$ as well as $\rho^{\otimes m}$ for some quantum state~$\rho \in \C^{d \times d}$, it outputs a matrix $\wh{\brho} \in \C^{d \times d}$ (not necessarily a state) satisfying $\E[\|\rho - \wh{\brho}\|_\Fro^2] \leq f(\rho)/m$.
\end{definition}
\Cref{thm:ow,thm:krt} may be restated as follows:
\begin{theorem} \label{thm:ow2}
    There is an estimation algorithm with Frobenius-squared rate~$O(d)$ on $d$-dimensional states.
\end{theorem}
\begin{theorem} \label{thm:krt2}
    There is an estimation algorithm using single-copy measurements with Frobenius-squared rate~$O(rd)$ on $d$-dimensional states of rank at most~$r$.
\end{theorem}
Finally, \Cref{prop:simple} gives a simple single-copy measurement algorithm that has Frobenius-squared rate~$O(d^2)$ (matching matching \Cref{thm:krt2} in the high-rank case).\\

We will now successively describe several black-box ``upgrades'' one may make to a Frobenius-squared estimation algorithm.
\emph{All of these will have the feature that they preserve the single-copy measurement property.} 
Our ultimate goal will be to upgrade to closeness guarantees with respect to much stronger distance measures, with minimal loss in rate.
To illustrate the idea, we start with a very simple upgrade (that most natural algorithms are unlikely to need):
\begin{proposition} \label{prop:hermit}
    A Frobenius-squared estimation  algorithm may be transformed to one that always outputs Hermitian estimates, with no loss in rate.
\end{proposition}
\begin{proof}
    Given algorithm~$\calA$, let~$\calA'$ be the algorithm that on input~$\rho$ runs~$\calA$, producing~$\wh{\brho}$, and then outputs $\wh{\brho}_H \coloneqq (\wh{\brho} + \wh{\brho}^\dagger)/2$, so that $\wh{\brho}_A \coloneqq (\wh{\brho} - \wh{\brho}^\dagger)/2$ and $\wh{\brho} = \wh{\brho}_H+\wh{\brho}_A$. 
    The Hermitian matrices are a real vector space, so by picking a (Hilbert--Schmidt) orthogonal basis (for example, the generalized Pauli matrices), is it easy to verify that for Hermitian $\rho$, we always have
    $\|\wh{\brho} - \rho\|_\Fro^2 = \|\wh{\brho}_H - \rho\|_\Fro^2 + \|\wh{\brho}_A\|_\Fro^2 \geq \|\wh{\brho}_H - \rho\|_\Fro^2$. 
    The claim then follows by taking expectations.
\end{proof}

The next upgrade is not a change in algorithm, but rather in terminology.
\begin{definition} \label{def:diag}
    Say that an  estimation algorithm with Frobenius-squared rate~$f$ \emph{returns diagonal estimates} if, when run on~$\rho \in \C^{d \times d}$, it returns a unitary~$\bU$ and a (real) diagonal matrix~$\wh{\brho} = \diag(\bq)$ with $\bq_1 \leq \bq_2 \leq \cdots \leq \bq_d$  such that $\E[\|\bU \rho \bU^\dagger - \wh{\brho}\|_\Fro^2] \leq f(\rho)$.
\end{definition}
Given such an algorithm, we can get an Frobenius-squared estimator with rate~$f$ for~$\rho$ just by returning $\bU^\dagger \wh{\brho} \bU$. 
But we will prefer the interpretation that the algorithm is allowed to ``revise''~$\rho$ to state $\bU \rho \bU^\dagger$ (with $\bU$ of its choosing), and then try to estimate this new state.
\begin{proposition} \label{prop:dig}
    A Frobenius-squared estimation  algorithm may be transformed to one that returns diagonal estimates, with no loss in rate.
\end{proposition}
\begin{proof}
    First we transform the algorithm to output Hermitian estimates, using \Cref{prop:hermit}.
    Then, given output~$\wh{\brho}$, 
    the algorithm simply chooses a unitary~$\bU$ such that $\wh{\brho} = \bU \diag(\bq) \bU^\dagger$ with $\bq_1 \leq \cdots \leq \bq_d$, and returns the unitary~$\bU$ along with diagonal estimate~$\diag(\bq)$.  The proof is complete because Frobenius-squared distance is unitarily invariant.
\end{proof}

\begin{proposition} \label{prop:fixlearner}
    With only constant-factor rate loss, a Frobenius-squared estimation algorithm may be transformed to one that outputs diagonal estimates~$\wh{\brho} = \diag(\bq)$ that are genuine quantum states, meaning that~$\bq$ is a probability vector.
\end{proposition}
\begin{proof}
    First we apply \Cref{prop:dig}, obtaining algorithm~$\calA'$ with diagonal estimates and rate~$f$.
    Now our transformed algorithm, when given $m$~copies of~$\rho$, will start by running~$\calA'$  on the first~$m/2$ copies (we may assume~$m$ is even), yielding a diagonal estimate~$\brho'$ (and, to be formal, a unitary~$\bU$ which should be used to conjugate the remaining copies of~$\rho$).  Say that $\|\brho' - \rho\|_\Fro^2 = \boldsymbol{\eta}$, and recall that $\E[\boldsymbol{\eta}] \leq 2f(\rho)/m$.  The next step of the algorithm is to use single-copy standard basis measurements with the remaining~$m/2$ copies of~$\rho$ to make a new estimate~$\bq$ of the diagonal of~$\rho$.
    Applying \Cref{prop:learn-L2}, the  empirical estimator~$\bq$ is a genuine probability distribution, and the algorithm will finally output~$\wh{\brho} = \diag(\bq)$. (Actually, since $\bq$ might not have nondecreasing entries, we should finally ``revise'' by a permutation matrix.)
    The Frobenius-squared error of~$\wh{\brho}$ is its off-diagonal Frobenius-squared error plus its diagonal Frobenius-squared error; the former is at most~$\boldsymbol{\eta}$ and the latter is, in expectation, at most~$2/m$ by \Cref{prop:learn-L2}.  Since $\E[\boldsymbol{\eta}] \leq 2f(\rho)/m$, the total expected Frobenius-squared error is at most $2f(\rho)/m + 2/m = O(f(\rho))/m$, as needed.
\end{proof}

We will also need a high-probability version of the preceding result, with some extra properties.  The reader should recall the $m_\delta$ notation from \Cref{prop:learn-L2-high}.
\begin{proposition} \label{prop:subhigh}
    The algorithm from \Cref{prop:fixlearner} may be modified so that, given $0 < \delta < 1/2$, its output satisfies each of the following statements except with probability at most~$\delta$ (for any fixed $i \in [d]$):
    \begin{itemize}
        \item $\|\rho - \wh{\brho}\|_\Fro^2 \leq f/m_\delta$;
        \item if $\rho_{ii} \geq 1/m_\delta$ then $\wh{\brho}_{ii}$ is within a $1.01$-factor of $\rho_{ii}$;
        \item if $\rho_{ii} \leq 1/m_\delta$ then $\wh{\brho}_{ii} \leq 1.01/m_\delta$.
    \end{itemize}
\end{proposition}
\begin{proof}
    The first statement may be obtained in a black-box way using the ``median trick'', which  upgrades estimation-in-expectation  to estimation-with-confidence-$(1-\delta)$ at the expense of only an $O(\log(1/\delta))$ sample complexity factor.  This trick may be applied whenever the loss measure is a metric (as Frobenius distance is); see, e.g., \cite[Prop.~2.4]{HKOT23} for details.  It is sufficient to prove this statement with the~$O(\cdot)$, because we may  then remove it by raising~$c$ in the $m_\delta$ notation.  (Similarly, we may tolerate achieving $2\delta$ failure probabilities, rather than~$\delta$.)

    To get the other two conclusions, we need to re-estimate the diagonal of~$\rho$, just as we did in \Cref{prop:fixlearner}.
    For this we use \Cref{prop:learn-L2-high}.  As in \Cref{prop:fixlearner}, this re-estimation contributes some new on-diagonal Frobenius-squared distance, but  only at most~$1/m_\delta \leq f/m_\delta$; thus the proposition's first statement remains okay.
    The remaining statements follow from \Cref{prop:learn-L2-high} by taking its~``$S$'' to be~$\{i\}$.
\end{proof}

Now we come to a most important reduction: being able to estimate \emph{subnormalized states}.  
Let us define terms, and make the simplifying assumption that rate functions for proper states only depend on dimension and rank, and that they are nondecreasing functions of these parameters.
We also assume for simplicity that our subnormalized states arise just from submatrices, but they could just as well arise from any given projector~$\Pi$.
\begin{definition}
    We say a \emph{subnormalized} state estimation algorithm~$\calA$ has Frobenius-squared rate~$f(d,r,\tau)$ if the following holds:
    Whenever $\calA$ is given a subset $S \subseteq [d]$, as well as $\rho^{\otimes m}$ for some quantum state~$\rho \in \C^{d \times d}$ of rank at most~$r$, it outputs an estimate $\wh{\brho}[S] \in \C^{S \times S}$ such that $\E[\|\rho[S] - \wh{\brho}[S]\|_\Fro^2] \leq f(d,r,\tau)/m$, where $\tau$ denotes $\tr \rho[S]$. 
\end{definition}
\begin{remark}
    In the above definition, we may also include the condition of ``returning diagonal estimates'' as in \Cref{def:diag}, with the returned unitary~$\bU$ being in~$\C^{S \times S}$.
    Moreover, for linguistic simplicity we will henceforth assume that ``diagonal estimates'' are also required to have nonnegative (diagonal) entries.
\end{remark} 
\begin{remark}
    Our subnormalized state estimation algorithms will actually achieve improved rate $f(d',r',\tau)$, where $d' = |S| \leq d$ and $r' = \rank \rho[S] \leq r$, but we will not try to squeeze anything out of this, for simplicity.
\end{remark}
\begin{proposition} \label{prop:subby}
    A state estimation algorithm  $\calA$ with Frobenius-squared rate~$f(d,r)$ may be transformed to a subnormalized state estimation algorithm~$\calA'$, returning diagonal estimates, and having Frobenius-squared rate $f(d,r,\tau) = O(\tau \cdot f(d,r))$.
\end{proposition}
\begin{proof}
    We first apply \Cref{prop:fixlearner} so that~$\calA$ may be assumed to output diagonal, genuine quantum states. 
    This only changes bounds by constant factors on~$m$, to which the statement of this proposition is anyway insensitive. 
    
    Given $S \subseteq [d]$ and $\rho^{\otimes m}$, let us write $\tau = \tr \rho[S]$ and also introduce the quantum state $\rho_{\mid S} = \rho[S]/\tau$ (when $\tau > 0$).
    The first step of the new algorithm~$\calA'$ is to measure each copy of~$\rho$ using the two-outcome PVM~$(\Id_S, \Id_{[d] \setminus S})$.
    It retains all copies that have outcome~$S$ and discards the rest.
    In this way, $\calA'$ obtains $(\rho_{\mid S})^{\otimes \bm'}$, where $\bm' \sim \textrm{Binomial}(m, \tau)$.
    If $\bm' = 0$ then the algorithm will return the $0$~matrix. 
    Otherwise, if $\bm' \neq 0$ the algorithm applies $\calA$ to $\rho_{\mid S}$ and obtains an estimate $\wh{\brho}_{\mid S}$ with expected Frobenius-squared error at most $f(d',r')/{\bm'} \leq f(d,r)/{\bm'}$, where $d' = |S|$, $r' = \rank \rho[S]$.
    The final estimate that~$\calA'$ produces for~$\rho[S]$ will be $\wh{\brho}[S] \coloneqq (\bm'/m) \wh{\brho}_{\mid S}$; indeed, we can use this expression even in the $\bm' = 0$ case. 
    We now have
    \begin{equation}    \label{eqn:yuk1}
        \E\bigl[\|\rho[S] - \wh{\brho}[S]\|^2_\Fro\bigr] = 
        \E\bigl[\|\tau \rho_{\mid S} - (\bm'/m) \wh{\brho}_{\mid S}\|^2_\Fro\bigr] = 
        \tau^2 \E\bigl[\|\rho_{\mid S} - (\bm'/(\tau m)) \wh{\brho}_{\mid S}\|^2_\Fro\bigr].
    \end{equation}
    We write
    \begin{equation} \label{eqn:yuk2}
        \rho_{\mid S} - (\bm'/(\tau m)) \wh{\brho}_{\mid S} = \bDelta + \bR, \quad \bDelta \coloneqq \rho_{\mid S} - \wh{\brho}_{\mid S}, \quad \bR \coloneqq (1-\bm'/(\tau m)) \wh{\brho}_{\mid S},
    \end{equation}
    and use 
    \begin{equation}    \label{eqn:yuk3}
        \E\bigl[\|\bDelta + \bR\|_\Fro^2\bigr] \leq 
        2\E\bigl[\|\bDelta\|_\Fro^2\bigr] +   2\E\bigl[\|\bR\|_\Fro^2\bigr].
    \end{equation}
    By assumption on~$\calA$, for $m' > 0$ we have
    \begin{equation}
        \E\bigl[\|\bDelta\|_\Fro^2 \mid \bm' = m'\bigr] \leq f(d,r)/m' \leq 2f(d,r)/(m'+1),
    \end{equation}
    and this is also true even for $m' = 0$ (recall we always assume $f \geq 1$).
    Using the elementary fact $\E\bigl[\frac{1}{\mathrm{Bin}(m,\tau) + 1}\bigr] = \frac{1 - (1-\tau)^{m+1}}{\tau(m+1)} \leq 1/(\tau m)$, we conclude
    \begin{equation}
        \E\bigl[\|\bDelta\|_\Fro^2\bigr] \leq 2f(d',r')/(\tau m).
    \end{equation}
    As for $\E\bigl[\|\bR\|_\Fro^2\bigr]$, let us first observe that conditioned on any $\bm' = m$ (including $m = 0$), we have $\|\wh{\brho}_{\mid S}\|_\Fro \leq 1$ with certainty, simply because $\calA$ always outputs a genuine quantum state. 
    Thus
    \begin{equation}
        \E\bigl[\|\bR\|_\Fro^2\bigr] \leq \E\bigl[(1-\bm'/(\tau m))^2\bigr] = (1-\tau)/(\tau m) \leq 1/(\tau m).
    \end{equation}
    Combining all of the above (and using $f \geq 1$ again), we conclude
    $\E\bigl[\|\rho[S] - \wh{\brho}[S]\|^2_\Fro\bigr] \leq \tau^2 \cdot O(f(d,r)/(\tau m))$, as needed.
\end{proof}

Finally, we use \Cref{prop:subhigh} to obtain some high-probability guarantees:
\begin{proposition} \label{prop:final-upgrade}
    A state estimation algorithm having Frobenius-squared rate~$f(d,r)$ may be transformed (preserving the single-copy measurement property) into a subnormalized state estimation algorithm returning diagonal estimates with the following properties:  
    
    Given parameters~$r$,~$\delta$, and $S \subseteq [d]$, as well as $\rho^{\otimes m}$ for some quantum state $\rho \in \C^{d \times d}$ of rank at most~$r$, 
    the algorithm outputs a number~$\wh{\btau}$ and a (diagonal) estimate $\wh{\brho}[S] \in \C^{S \times S}$ such that, writing $\tau = \tr \rho[S]$ and recalling the notation $m_{\delta} = m/(c \ln(1/\delta))$ 
    (where $c \geq 1$ is some universal constant), we have the following:
    \begin{enumerate}[label=(\roman*)]
        \item \label{enum:c} if $\tau \leq 1/m_\delta$ then $\wh{\btau} \leq 1.1/m_\delta$ except with probability at most~$\delta$;
        \item \label{enum:z} if $\wh{\btau} \leq 1.1/m_\delta$, then $\|\rho[S] - \wh{\brho}[S]\|_{\Fro}^2 \leq O(\tau \cdot f(d,r)/m)$ except with probability at most~$.0001$;
    \item[] and if $\tau \geq 1/m_\delta$ then the following hold:
        \item \label{enum:b} the quantities $\tau$, $\wh{\btau}$, and $\tr \wh{\brho}[S]$ are all within a $1.1$-factor, except with probability at most~$\delta$;
        \item \label{enum:a} $\|\rho[S] - \wh{\brho}[S]\|_\Fro^2 \leq \tau \cdot f(d,r)/m_\delta$, except with probability at most~$\delta$;
        \item \label{enum:d} simultaneously for all $i \in S$ with $\rho_{ii} \geq \theta \coloneqq \max\{\tau/(100r), 1/m_{\delta/d}\}$, we have that $\wh{\brho}_{ii}$ is within a $1.1$-factor of~$\rho_{ii}$, except with probability at most~$\delta$.
        \item \label{enum:e} simultaneously for all $i \in S$ with $\rho_{ii} \leq \theta$, we have that $\wh{\brho}_{ii} \leq 1.1\theta$, except with probability at most~$\delta$.
    \end{enumerate}
\end{proposition}
\begin{proof}
    Since the definition of~$m_\delta$ anyway contains an unspecified constant~$c$, it is sufficient to prove the proposition with constant losses on various bounds (and then raise~$c$'s value to compensate).  In particular, for notational simplicity we assume that we get~$2m$ rather than~$m$ copies of~$\rho$.
    
    The algorithm begins by using the first $m$~copies of~$\rho$ to obtain $(\rho_{\mid S})^{\otimes \bm'}$ as in \Cref{prop:subby}; this is done just to get~$\bm'$. The algorithm's output~$\wh{\btau}$ is~$\bm'/m$, and the proposition's conclusion \Cref{enum:c} follows straightforwardly from  Chernoff bounds
    (assuming~$c$ is sufficiently large).
    Similarly, the $\tau$-vs.-$\wh{\btau}$ part of \Cref{enum:b} follows from a Chernoff bound, and we will actually ensure $1.01$-factor closeness for later convenience. 
    
    If $\wh{\btau} \leq 1.1/m_\delta$, then the algorithm runs \Cref{prop:subby} on the second $m$~copies of~$\rho$, outputting the result.
    The conclusion in \Cref{enum:z} then holds except with probability at most~$.0001$, by applying Markov's inequality to \Cref{prop:subby}'s guarantee.

    We now describe how the remainder of the algorithm proceeds, when $\wh{\btau} \geq 1.1/m_\delta$.
    Note that since it only remains to prove \Cref{enum:b,enum:a,enum:d,enum:e}, we may as well assume $\tau \geq 1/m_{\delta}$.  
    The algorithm proceeds similarly to \Cref{prop:subby}, using the second $m$~copies of~$\rho$ to get $(\rho_{\mid S})^{\otimes \bm'}$ for a new value of~$\bm'$.  Since we are now assuming $\tau \geq 1/m_\delta$, a Chernoff bound implies that except with probability at most~$\delta$ we'll have
    \begin{equation}    \label{eqn:factory}
        \abs{\frac{\bm'}{m} - \tau} \leq .01\sqrt{\tau/m_\delta}  \leq .01 \tau \quad\implies\quad \frac{\bm'}{m} \text{ is within a $1.02$-factor of } \tau
    \end{equation}
    (as always, assuming $c$ is large enough).
    The algorithm now applies \Cref{prop:subhigh} in place of \Cref{prop:subby}, getting an estimate~$\wh{\brho}_{\mid S}$ of~$\rho_{\mid S}$ that satisfies the conclusions of \Cref{prop:subhigh}.  Finally, as before, the algorithm produces $\wh{\brho}[S] \coloneqq (\bm'/m)\wh{\brho}_{\mid S}$ as its final estimate.
    Let us now verify \Cref{enum:b,enum:a,enum:d,enum:e}.

    First, $\tr \wh{\brho}[S] = \bm'/m$, which by \Cref{eqn:factory} is within a $1.02$-factor of~$\tau$, thereby completing the proof of \Cref{enum:b} (recall that $\tau$~and~$\wh{\btau}$ are within a $1.01$-factor).

    Next, we verify \Cref{enum:a} up to a constant factor (as is sufficient).
    Following \Cref{eqn:yuk1,eqn:yuk2,eqn:yuk3} (but without expectations), we have
    \begin{equation}
        \|\rho[S] - \wh{\brho}[S]\|_{\Fro}^2 \leq \tau^2 (2 \|\rho_{\mid S} - \wh{\brho}_{\mid S}\|_{\Fro}^2 + 2\|\bR\|_{\Fro}^2) \leq 2\tau^2 \cdot f(d,r)/{\bm'_{\delta}} + 2(\tau-\bm'/m)^2.
    \end{equation}
    But \Cref{eqn:factory} (and using $f \geq 1$) we can bound the above by $2.03 \tau \cdot f(d,r)/m_{\delta}$, establishing \Cref{enum:a}.

    To show \Cref{enum:d}, let $B$ denote the set of all $i \in S$ with $\rho_{ii} \geq \theta$. 
    Since $\theta \geq \tau/(100r) = (\tr \rho[S])/(100 r)$, we know that $|B| \leq 100r$. Moreover, for any $i \in B$ we may use
    \begin{equation}
        (\rho_{\mid S})_{ii} \geq \theta/\tau \geq 1/(\tau m_{\delta/d}) \geq 1/(\tau m_{\delta/r}) \geq 1/(1.02 \bm'_{\delta/r}),
    \end{equation}
    (employing \Cref{eqn:factory}).
    (We weakened $\delta/d$ to $\delta/r$ just to illustrate this is all we need for \Cref{enum:d}.) 
    So by using the second bullet point of \Cref{prop:subhigh} in a union bound over the at most $O(r)$ indices in~$B$, we conclude that (except with probability at most~$O(\delta)$) for all $i \in B$  it holds that $(\wh{\brho}_{\mid S})_{ii}$ is within a $1.01$-factor of $(\rho_{\mid S})_{ii}$, and hence (by \Cref{eqn:factory}) $\wh{\brho}_{ii}$ is within a $1.1$-factor of $\rho_{ii}$.  This completes the verification of \Cref{enum:d}.

    Finally, verifying \Cref{enum:e} is similar; for simplicity, we just union-bound over all $i \in S \subseteq [d]$, using the fact that~$\theta \geq 1/m_{\delta/d}$.
\end{proof}

\subsection{The plan for learning in~\texorpdfstring{$\chi^2$}{chi squared}: refining diagonal estimates on submatrices}  \label{sec:sketch}
Suppose we have come up with a diagonal estimate~$\sigma_1$ of~$\rho \in \C^{d \times d}$ having some Frobenius-squared distance $\eta_1 = \|\rho -\sigma_1\|_\Fro^2$. 
(Here we will have ``revised'' some original~$\rho$ by the unitary that makes~$\sigma_1$ diagonal; this revision will be taken into account in all future uses of~$\rho$.)
Suppose we now choose some~$d_2 \leq d_1 \coloneqq d$, define $\rho_2$ to be the top-left~$d_2 \times d_2$ submatrix of~$\rho$, and apply \Cref{prop:final-upgrade} to it.  The idea is that we hope to improve the top-left part of our estimate~$\sigma_1$.

Recall that \Cref{prop:final-upgrade} affords us a diagonal estimate~$\sigma_2 \in \C^{d_2 \times d_2}$; let us understand a little more carefully what this means. 
The algorithm will give us a unitary~$U_2 \in \C^{d_2 \times d_2}$ such that $\|U_2 \rho_2 U_2^\dagger - \sigma_2\|_\Fro^2 = \eta_2$ for some small value~$\eta_2$. 
The idea now is to ``revise''  both~$\rho_1 \coloneqq \rho$ and~$\sigma_1$ by the unitary $U_2 \oplus \Id$, where here $\Id$ has dimension~$d_1-d_2$.
By design, the revised version of~$\rho_2$ will have Frobenius-squared distance $\eta_2$ from~$\sigma_2$.
Moreover, after revision, the fact that $\|\rho_1 -\sigma_1\|_\Fro^2 = \eta_1$ is unchanged (since Frobenius distance is unitarily invariant). 
On the other hand, although~$\sigma_1$ was previously diagonal, it no longer will be after revision. 
But it's easy to see that it will remain diagonal \emph{except} on its top-left $d_2 \times d_2$ block, which we are intending to replace by~$\sigma_2$ anyway. 
In particular, the \emph{off-diagonal} $d_2 \times (d_1-d_2)$ and $(d_1-d_2) \times d_2$ blocks of~$\sigma_1$ remain zero.

Let us summarize. 
We will first obtain a diagonal estimate $\sigma_1$ of~$\rho_1$ with some error~$\eta_1$. 
Then after choosing some $d_2 \in [d_1]$, we will obtain a further diagonal estimate~$\sigma_2$ of the top-left $d_2 \times d_2$ block of~$\rho$, with some error~$\eta_2$. 
We might then take as final estimate~$\wh{\rho}$ the diagonal matrix formed by replacing the top-left $d_2 \times d_2$ block of~$\sigma_1$ by~$\sigma_2$.

Naturally, this plan can be iterated (meaning we can try to improve the estimate's top-left $d_3 \times d_3$ block for some $d_3 \in [d_2]$) but let us pause here to discuss error.  
If we're interested in the Frobenius-squared error of our current estimate~$\wh{\rho}$, we can't say more than that it is bounded by~$\eta_1 + \eta_2$. 
Here we're decomposing the error into the contribution from the top-left $d_2 \times d_2$ block (which is $\eta_2$) plus the contribution from the remaining \raisebox{.3ex}{$\pmb{\lrcorner}$}-shaped region (consisting of the bottom-right $(d_1-d_2) \times (d_1-d_2)$ block plus the two off-diagonal blocks). 
We will just bound this second error contribution by the whole Frobenius-squared distance of $\sigma_1$~from~$\rho_1$, which is~$\eta_1$.

It would seem that this scheme of refining our estimate for the top-left block hasn't helped, since it took us from Frobenius-squared error $\eta_1$ to Frobenius-squared error (at most) $\eta_1 + \eta_2$. 
But the idea is that our new estimate~$\wh{\rho}$ may have improved \emph{Bures $\chi^2$-divergence}.
Recall the formula for $\chi^2$-divergence, \Cref{eqn:chi-formula} (which we will apply even though~$\wh{\rho}$ might not precisely be a state, meaning of trace~$1$).
Recall also that our diagonal estimates $\sigma_1 = \diag(q^{(1)})$ and $\sigma_2 = \diag(q^{(2)})$ are chosen to have nondecreasing entries along the diagonal. 
(We moreover expect that~$\wh{\rho}$ will also have nondecreasing entries, meaning $q^{(2)}_{d_2} \leq q^{(1)}_{d_2+1}$, but we won't rely on this.)
Now we can use the bound
\begin{multline}
    \DBchi{\rho}{\wh{\rho}} \leq \sum_{i,j=1}^d \frac{2}{q_{\max(i,j)}} |\rho_{ij} - \wh{\rho}_{ij}|^2 
    \leq \frac{2}{q^{(2)}_1} \sum_{i,j=1}^{d_2} |\rho_{ij} - \wh{\rho}_{ij}|^2 +  \frac{2}{q^{(1)}_{d_2+1}} \sum_{\substack{i,j : \\ \max(i,j) > d_2}} |\rho_{ij} - \wh{\rho}_{ij}|^2  \\
    \leq \frac{2}{q^{(2)}_1} \eta_2 + \frac{2}{q^{(1)}_{d_2+1}} \eta_1.
\end{multline}
The idea here is that if, perhaps
\begin{equation} \label{eqn:opti}
    q^{(2)}_1 \approx \cdots q^{(2)}_{d_2} \approx (\tr \sigma_2)/r; \quad \text{and} \quad q^{(1)}_{d_2+1} \approx \cdots \approx q^{(1)}_{d} \approx (\tr \sigma_1)/r,
\end{equation}
then hopefully from \Cref{prop:subby} with $m$~copies we will have $\eta_1 \approx O(\tau_1 \cdot f(d,r))/m$ and $\eta_2 \approx O(\tau_2 \cdot f(d,r))/m$, where $\tau_i = \tr \rho_i \approx \tr \sigma_i$. 
Then the total $\chi^2$-error would be approximately
\begin{equation}
    \frac{2r}{\tau_2} \cdot O(\tau_2 \cdot f(d,r))/m + \frac{2r}{\tau_1} \cdot O(\tau_1 \cdot f(d,r))/m = O(r \cdot f(d,r)) / m.
\end{equation}
This would mean we have converted Frobenius-squared rate $O(f(d,r))$ to $\chi^2$-divergence rate $O(r \cdot f(d,r))$.
Now \Cref{eqn:opti} might seem a little optimistic, but our idea will be that no matter what $\rho$'s eigenvalues are, we can break them up into logarithmically many groups where they only differ by a constant factor, and thereby achieve the desired $\chi^2$-divergence rate of $O(r \cdot f(d,r))$ up to logarithmic losses. 
Unfortunately, we will have to deal separately with any extremely small eigenvalues of~$\rho$, which causes some additional losses.

\begin{remark}  \label{rem:itshard}
    Ideally this plan suggests we might be able to achieve sample complexity $n = \wt{O}(rd/\eps)$ for tomography with respect to Bures $\chi^2$-divergence (for collective measurements).  
    But the ``small eigenvalue'' issue causes problems for this.
    Without explicitly claiming a lower bound, let us sketch why it seems difficult to significantly beat the $n = \wt{O}(r^{.5}d^{1.5}/\eps)$ complexity from \Cref{cor:1}, even in the case $r = 1$.  
    
    So suppose~$\rho = \ketbra{v}{v}$ for an unknown unit vector $\ket{v} \in \C^d$.
    The best known tomography algorithm for a pure state is extremely natural and simple~\cite{Hayashi1998}; it outputs a pure state $\ketbra{\bu}{\bu}$ and achieves $\abs{\braket{\bu}{\bv}} \geq 1-\eta$, i.e.\ infidelity~$\eta$, with high probability using $n = O(d/\eta)$ copies.
    Moreover, $n = \Omega(d/\eta)$ is a known lower bound~\cite{haah2017sample}.  
    However, with certainty we will have $\abs{\braket{\bu}{\bv}} \neq 1$ and hence $\DBchi{\rho}{\ketbra{\bu}{\bu}} = \infty$.
    To achieve $\DBchi{\rho}{\wh{\brho}} < \infty$ we will have to output a full-rank hypothesis~$\wh{\brho}$, and to achieve $\DBchi{\rho}{\wh{\brho}} \leq O(\eps)$ it's hard to imagine what to try besides something like $\wh{\brho} = \Delta_{\eps}(\ketbra{\bu}{\bu})$.
    But with this choice it's not hard to compute that $\DBchi{\rho}{\wh{\brho}} \geq (d/\eps)\cdot \Omega(\eta^2)$, seemingly forcing us to choose $\eta = \Theta(\eps/\sqrt{d})$ and thereby use $n = \Omega(d^{1.5}/\eps)$ copies.
\end{remark}

\subsection{The central estimation algorithm}   \label{sec:central}

\begin{theorem} \label{thm:central}
    A state estimation algorithm $\calA$ having Frobenius-squared rate~$f = f(d,r)$ satisfying\footnote{This mild assumption is made to keep parameter-setting simpler.} $f \gg \log d$ may be transformed (preserving the single-copy measurement property) into a state estimation algorithm~$\calA'$ returning diagonal estimates with the following properties:  

    Given $r$ and~$m \geq r$, the algorithm sets the following parameters:
    \begin{align*}
        \delta &= \frac{.0001}{\log_2(m/rf)}, &\tilde{\eps} &= C r f/m_\delta,  &\ell_{\max} &= \lceil\log_2(1/\tilde{\eps}) \rceil, &\eps &= \tilde{\eps} \ell_{\max} &M &= 2m \ell_{\max}.
    \end{align*}
    (Here $C$ is a large universal constant, and it is assumed that~$\eps$ is at most some small universal constant.) 
    
    Then, given $\rho^{\otimes M}$, where $\rho \in \C^{d \times d}$ is a quantum state of rank at most~$r$, 
    the algorithm~$\calA'$ outputs (with probability at least~$.99$) a partition $[d] = \bL \sqcup \bR$ (with $\bL = [\bd']$ for some $\bd' \leq d$), together with a quantum state $\wt{\brho} = \diag(\bq) \in \C^{d \times d}$ satisfying:
    \begin{enumerate} [label=(\alph*)]
        \item $|\bR| \leq O(r \ell_{\max})$; \label{enum:0}
        \item $\btau, \beps' \leq O(\tilde{\eps})$, where  $\btau \coloneqq \tr \rho[\bL]$ and $\beps' \coloneqq \tr \wt{\brho}[\bL]$; \label{enum:A}
        \item $\|\rho[\bL] - \wt{\brho}[\bL]\|_{\Fro}^2 \leq O(\frac{\tilde{\eps}^2}{r \ln(1/\delta)}) \leq O(\tilde{\eps}^2/r)$; \label{enum:B}
        \item $\DBchiminus{\bL}{\rho}{\wt{\brho}} \leq O(\eps)$. \label{enum:C}
    \end{enumerate}
\end{theorem}
\begin{proof}
    Fix a Frobenius-squared estimation algorithm~$\calA$ with rate~$f=f(d,r)$, and assume we have passed it through \Cref{prop:final-upgrade} so that we may use it to make diagonal estimates of subnormalized states.

    The algorithm $\calA'$ will run in some~$\bell$  stages, where we guarantee $\bell \leq \ell_{\max}$. Each stage will consume~$m$ copies of~$\rho$. After the $\bell$th stage, there will be some final processing that uses the remaining~$M/2$ (at least) copies of~$\rho$.
    
    As the algorithm progresses, it will define a sequence of numbers $d = \bd_1 \geq \bd_2 \geq \cdots \geq \bd_{\bell}$, with the value~$\bd_{t+1}$ being selected at the end of the $t$th stage. 
    We introduce the notation $\bR_t = \{\bd_{t+1} + 1, \dots, \bd_t\}$; each of these sets will have cardinality at most~$r$.

    At the beginning of the $t$th stage, $\calA'$~will run the algorithm from \Cref{prop:final-upgrade} on~$\rho[\bd_t]$, with confidence parameter~$\delta$, resulting in some $\wh{\btau}_t$ and a diagonal estimate that we will call~$\bsigma_t$.
    We will use the fact that
    $\delta$ always satisfies all of the following (provided~$C$ is large enough and using $f \geq \log d$):
    \begin{equation}    \label[ineq]{ineq:del}
        1/m_{\delta} \leq \tilde{\eps}, \quad 1/m_{\delta/d} \ll \tilde{\eps}/r, \quad \delta \leq .0001/\ell_{\max}.
    \end{equation}
    By losing probability at most~$5\delta$ in each stage, we may assume that except with probability at most~$.0006$, all of the~desired outcomes from \Cref{prop:final-upgrade} do occur over the course of the algorithm.
    
    If $\wh{\btau}_t \leq 1.1 \tilde{\eps}$ or $t > d$, then this is declared the final stage; i.e., the algorithm will define $\bell = t$ and move to its ``final processing''.  
    Otherwise, in a non-final stage we have $\wh{\btau} > 1.1 \tilde{\eps} \geq 1.1/m_\delta$ (using \Cref{ineq:del}), so by \Cref{enum:c} of \Cref{prop:final-upgrade} we may assume that $\tr \rho[\bd_t] > 1/m_\delta$, and hence (using \Cref{enum:b}),
    \begin{equation}    \label[ineq]{ineq:nonf}
        \text{for $t < \bell$}, \qquad \tr \bsigma_t[S] \text{ is within a $1.1$-factor of } \tau_t \coloneqq \tr \rho[\bd_t]; \qquad \text{moreover, } \tau_t \geq \tilde{\eps}.
    \end{equation}
    Next, using \Cref{enum:d} of \Cref{prop:final-upgrade} and $1/m_{\delta/d} \ll \tilde{\eps}/r \leq \tau_t/r$ (which implies that the Proposition's ``$\theta$''~is~$\tau/(100r)$), we get that
    \begin{equation}    \label[ineq]{ineq:llama}
        \text{for $t < \bell$}, \quad \text{for all $i \leq \bd_t$ with $\rho_{ii} \geq \tau_t/(100r)$}, \quad (\bsigma_t)_{ii} \text{ is within a $1.1$-factor of } \rho_{ii}.
    \end{equation}
    Moreover, from \Cref{enum:e} we get
    \begin{equation}    \label[ineq]{ineq:chicken}
        \text{for $t < \bell$}, \quad \text{for all $i \leq \bd_t$ with $\rho_{ii} \leq \tau_t/(100r)$}, \quad (\bsigma_t)_{ii} \leq 1.1\tau_t/(100 r).
    \end{equation}
    Finally, we record the main conclusions \Cref{enum:z,enum:a} of \Cref{prop:final-upgrade}, taking care to distinguish the final stage:
    \begin{equation}    \label[ineq]{ineq:frb}
        \text{for $t <\bell$}, \quad \|\rho[\bd_t] - \bsigma_t\|_\Fro^2 \leq \tau_t  f/m_\delta; \qquad\text{and} \qquad \text{for $t = \bell$}, \quad \|\rho[\bd_t] - \bsigma_t\|_\Fro^2 \leq O(\tau_{\bell}  f/m).
    \end{equation}

    Now we explain how algorithm $\calA'$ defines~$\bd_{t+1}$ at the end of non-final stage~$t$, where recall non-finality implies from \Cref{ineq:nonf}
    \begin{equation}    \label[ineq]{ineq:barb}
        \tau_t \geq \tilde{\eps} = C r f / m_\delta.  
    \end{equation}
    Considering the first bound in \Cref{ineq:frb}, note that $\rank \rho[\bd_t] \leq r$, so we have that the diagonal matrix~$\bsigma_t$ has Frobenius-squared distance at most $\tau_t f/m_\delta$ from a matrix of rank at most~$r$.  
    But the rank-at-most-$r$ matrix that is Frobenius-squared-closest to~$\bsigma_t$ is simply $\bsigma'_t$, the matrix formed by zeroing out all but the~$r$ largest entries of~$\bsigma_t$.  
    Recalling that $\bsigma'_t$ has nondecreasing diagonal entries, this means $\bsigma'_t$ is formed by zeroing out all diagonal entries of index at most $\bd'_{t+1} \coloneqq \max\{\bd_t - r,0\}$. 
    Thus we have
    \begin{equation}    \label[ineq]{ineq:chef}
        \|\rho[\bd_t] - \bsigma'_t\|_\Fro^2 \leq 4\tau_t  f/m_\delta \quad\implies\quad \|\rho[\bd'_{t+1}]\|_\Fro^2 \leq 4\tau_t  f/m_\delta \quad \implies \quad (\tr \rho[\bd'_{t+1}])^2 \leq r \cdot 4 \tau_t f/m_\delta,
    \end{equation}
    where the last deduction used $\rank \rho[\bd'_{t+1}] \leq r$.
    But assuming $C \geq 64$, \Cref{ineq:barb} implies
    \begin{equation}
        r \cdot 4 \tau_t f/m_\delta = \tau_t \cdot (4r f/m_\delta) \leq \tau_t \cdot \tfrac{1}{16} \tau_t = (\tfrac{1}{4} \tau_t)^2,
    \end{equation}
    Thus from \Cref{ineq:chef} we conclude
    \begin{equation}
        \tr \rho[\bd'_{t+1}] \leq \tfrac14 \tau_t
          \quad \implies \quad \tr \rho[\bR'_{t}] \geq \tfrac34\tau_t \text{ for $\bR'_{t} \coloneqq \{\bd'_{t+1} + 1, \dots, \bd_t\}$}.
    \end{equation}
    Let $\bR''_{t} = \{i \in \bR'_{t} : \rho_{ii} > (1.1)^4\tau_t/(100r) = .014641 \tau_t/r\}$. Since $\abs{\bR'_{t}} \leq r$, the sum of $\rho_{ii}$ over all $\bR'_{t} \setminus \bR''_{t}$ is at most $.014641\tau_t \leq .02 \tau_t$; hence 
    \begin{equation}    \label[ineq]{ineq:cats}
        \sum_{i \in \bR''_{t+1}} \rho_{ii} \geq .73 \tau_t, \quad \text{and each summand  exceeds $(1.1)^4\frac{\tau_t}{100r}$}.
    \end{equation} 
    Applying \Cref{ineq:nonf,ineq:llama}, we conclude that 
    \begin{equation}    \label[ineq]{ineq:return}
        \sum_{i \in \bR''_{t}} (\bsigma_t)_{ii} \geq \tfrac{.73}{1.1} \tau_t > .66\tau_t, \quad \text{and each summand exceeds $(1.1)^3\frac{\tau_t}{100r} \geq (1.1)^2\frac{\tr \bsigma_t}{100r}$}.
    \end{equation}
    We now stipulate that algorithm~$\calA'$ chooses $\bd_{t+1} \geq \bd'_{t+1}$ to be minimal so that
    \begin{equation} \label[ineq]{ineq:fuz}
        (\bsigma_t)_{ii} > (1.1)^2\frac{\tr \bsigma_t}{100r} \quad \text{for all $i \in \bR_{t} = \{\bd_{t+1} +1, \dots, \bd_t\}$}.
    \end{equation}
    (In other words, $\{\bd_{t+1} + 1, \dots, \bd_t\}$ is the maximum-cardinality suffix of $\bR'_t$ where the above holds.)
    Then
    \begin{equation}    \label[ineq]{ineq:ges}
        \sum_{i \in \bR_{t}} (\bsigma_t)_{ii} = \tr \bsigma_t[\bR_{t}] > .66 \tau_t.
    \end{equation}
    Moreover, from \Cref{ineq:nonf,ineq:chicken} we know that $\rho_{ii} > \tau_t/(100 r)$ for all $i \in \bR_{t}$; hence 
    \begin{equation}    \label[ineq]{ineq:ballgame}
        (\sigma_t)_{ii} \text{ is within a $1.1$-factor of~$\rho_{ii}$ for all $i \in \bR_t$},
    \end{equation}
    and therefore \Cref{ineq:ges} implies $\tr \rho[\bR_{t}] > \frac{.66}{1.1} \tau_t \geq \tfrac12 \tau_t$, whence 
    \begin{equation}    \label[ineq]{ineq:drive}
        \tau_{t+1} = \tr \rho[\bd_{t+1}] = \tr \rho[\bd_t] - \tr \rho[\bR_{t}] \leq \tau_t - \tfrac12 \tau_t < \tfrac12 \tau_t.
    \end{equation}
    This is the key deduction that lets us make progress, in particular confirming that~$\bell \leq \lceil\log_2(1/\tilde{\eps})\rceil$ (because of \Cref{ineq:nonf}). \\

    Now we discuss the ``final processing''.
    The final partition output by~$\calA'$ will be $\bL \sqcup \bR$, where $\bL = [\bd_{\bell}]$ and $\bR = \bR_1 \sqcup \cdots \sqcup \bR_{\bell - 1}$; since $|\bR_t| \leq r$ for all~$t$ we satisfy the theorem's \Cref{enum:0}.
    We can verify the bound on~$\btau = \tr \rho[\bd_{\bell}]$ in \Cref{enum:A}  by recalling that when the final stage is reached we have $\wh{\btau} \leq 1.1 \tilde{\eps}$, and hence $\btau  \leq (1.1)^2 \tilde{\eps}$ by \Cref{enum:b} of \Cref{prop:final-upgrade} (recall $\tilde{\eps} \geq 1/m_\delta$).
    We can also partly verify the conclusion \Cref{enum:B} by observing that the \emph{off-diagonal} Frobenius-squared of~$\rho[\bL]$ is upper-bounded by $\|\rho[\bL] - \bsigma_{\bd_{\bell}}\|_\Fro^2$, and by \Cref{ineq:frb}  this is at most  $O(\tau_{\bell}  f/m) \leq O(\tilde{\eps} f/m) = O(\tilde{\eps}^2)/(r \ln(1/\delta))$.  Thus:
    \begin{equation} \label[ineq]{ineq:rum}
         \text{ \Cref{enum:B} holds provided } \left\|\diag(\rho)[\bL] - \wt{\brho}[\bL]\right\|_2^2 \leq O(\tilde{\eps}^2)/(r \ln(1/\delta)).
    \end{equation}
    
    Aside from establishing the above, it remains to describe how algorithm~$\calA'$ forms~$\wt{\brho}$ satisfying the theorem's conclusion \Cref{enum:C}.
    We first describe a candidate output we'll call~$\bsigma'$ that \emph{almost} works: namely, $\bsigma'$ is formed by setting its diagonal elements from~$\bR_t$ to be those from~$\bsigma_t$, for $t < \bell$. (The remaining diagonal entries may be set to~$0$.)
    The difficulty with this is that it's not easy to control $\tr \bsigma'$, but let us ignore this issue and calculate~$\chi^2$-divergence.\footnote{Strictly speaking, $\bsigma'$ need not have nondecreasing diagonal entries as promised, but we can finally  ``revise'' by a permutation matrix to fix this.}  Ignoring the fact that we are not working with normalized states, we may bound
    \begin{equation}
        \DBchiminus{\bL}{\rho}{\bsigma'} \leq \sum_{\substack{i,j \in [d] \\ k \coloneqq \max(i,j) \in \bR}} \frac{2}{\bsigma'_{kk}} \abs{\rho_{ij} - \bsigma'_{ij}}^2 
        \leq \sum_{t < \bell} \frac{2}{\min\{\bsigma'_{hh} : h \in \bR_t\}} \cdot \|\rho[\bd_t] - \bsigma_t\|_\Fro^2.
    \end{equation}
    Each summand above can be upper-bounded using \Cref{ineq:frb,ineq:fuz}, yielding
    \begin{equation}    \label[ineq]{ineq:33}
        \DBchiminus{\bL}{\rho}{\bsigma'} \leq  \sum_{t < \bell} \frac{2}{1.1 \tau_t/(100r)} \cdot \tau_t f/m_\delta \leq O( \ell_{\max} rf/m_\delta) =  O(\tilde{\eps} \ell_{\max}) = O(\eps).
    \end{equation}
    
    We now work to control the trace of our estimate.
    Our strategy is to have~$\calA'$ perform diagonal measurements on the  remaining $M/2$ copies of~$\rho$ to classically relearn its diagonal via \Cref{prop:learn-chi2}, with its ``$S$''~set to~$\bR$.
    Calling the resulting probability distribution~$\bq$, the algorithm will finally take $\wt{\brho} = \diag(\bq)$.

    First we complete the verification by \Cref{enum:B} by establishing the condition in \Cref{ineq:rum}: since~$\bq[\bL]$ is formed by the empirical estimator, Markov's inequality and \Cref{prop:learn-L2} imply that except with probability at most~$.0001$ we have $\left\|\diag(\rho)[\bL] - \wt{\brho}[\bL]\right\|_2^2 \leq O(\tr \rho[\bL])/(M/2) \leq O(\tilde{\eps}/(m \ell_{\max})) = O(\tilde{\eps}^2/(r \ln (1/\delta)f\ell_{\max})$, and we have a factor of $f\ell_{\max}$ to spare.
    
    Next, using Markov again with \Cref{prop:learn-chi2} we get that except with probability at most~$.0001$,
    \begin{equation}    \label[ineq]{ineq:34}        \dchisq{\diag(\rho[\bR])}{\bq[\bR]} \leq O(|\bR|/M) \leq O(r/m) \ll \tilde{\eps}.
    \end{equation}
    Also, using $f \geq \log d$ (and $C$ large enough), we indeed have $1/(M/2)_{\delta'} \leq \tilde{\eps}/(100r)$ for $\delta' = .0001/|\bR| \geq .0001/(r \ell_{\max})$; since also $\rho_{ii} \geq \tilde{\eps}/(100r)$ for  all $i \in \bR$ (recall \Cref{ineq:cats}), we conclude 
    \begin{equation}    \label[ineq]{ineq:virtue}
        \bq_i \text{ is within a $4$-factor of~$\rho_{ii}$ for all } i \in \bR,
    \end{equation}
    except with probability at most~$.0001$.
    Finally, it is easy to see that in \Cref{prop:learn-chi2} we have $\E[\| \bq[\bL]\|_1] \leq \tr \rho[\bL] = \btau \leq O(\tilde{\eps})$, and hence Markov implies that except with probability at most~$.0001$ we have $\beps' =         \|\bq[\bL]\|_1 \leq O(\tilde{\eps})$, completing the verification of \Cref{enum:A}.
    
    Finally we finish the analysis of $\DBchiminus{\bL}{\rho}{\wt{\brho}}$. The contribution to this quantity from the diagonal entries is precisely \Cref{ineq:34}. 
    On the other hand, since $\bsigma'$ and~$\wt{\brho}$ are both diagonal, the off-diagonal contribution to $\DBchiminus{\bL}{\rho}{\wt{\brho}}$ can be bounded by a constant times \Cref{ineq:33}, using the fact that the diagonal entries (from~$\bR$) of~$\bsigma'$ and~$\wt{\brho}$ are all within a constant factor by virtue of \Cref{ineq:ballgame,ineq:virtue}.  This completes the verification of \Cref{enum:C}.
\end{proof}

We also show the following,  to improve some $\log(1/\eps)$ factors in the case that $\eps$~is extremely small and $r = \Theta(d)$.  (The reader might like to think of the case when $d = O(1)$.)
\begin{theorem} \label{thm:central2}
    There is a variant version of~$\calA'$ from \Cref{thm:central} with the following alternative parameter settings:
    \begin{equation}
            \delta =\frac{.0001}{d+1}, \qquad \ell_{\max} = d+1, \qquad \eps = \frac{d \ln r}{r} \tilde{\eps}.
    \end{equation}
\end{theorem}
\begin{proof}
    Besides  verifying that \Cref{ineq:del} still holds with our changed $\delta$ and $\ell_{\max}$, there is one alternative idea to be explained.  
    In the preceding proof, the driver of progress was \Cref{ineq:drive} showing $\tau_{t+1} < \frac12 \tau_t$; this enabled us to take $\ell_{\max}$ logarithmic in~$1/\tilde{\eps}$.  In this variant, we will only use this inequality weakly, to show that $|\bR_t| \geq 1$ so that $d_{t+1} < d_t$ strictly; this is already enough to ensure that taking $\ell_{\max} = d+1$ is acceptable.  On the other hand, if we only implement this change then~$\eps$ would become unnecessarily large (namely, $\tilde{\eps} (d+1)$).

    To get the improved value of~$\eps$, we change how $\calA'$ chooses the $\bd_t$ values. 
    Returning to \Cref{ineq:return}, in the $t$th stage there is a set $\bR''_{t}$ of at most~$r$ indices~$i$ on which each $(\bsigma_t)_{ii}$ exceeds $\bbeta \coloneqq (1.1)^2 \cdot \frac{\tr \bsigma_t}{100r}$, and their sum~$\bs$ exceeds~$.66 \tau_t \geq .6(\tr \bsigma_t)$. 
    Then  $\calA'$ chooses $\bR_t$ to consist of all indices~$i \in \bR'_t$ with $(\bsigma_t)_{ii} \geq \bbeta$, of which there are at most~$O(r)$.  
    Note that if we conversely had $|\bR_t|$ at \emph{least}~$\Omega(r)$ for every~$t$, then the algorithm would halt in at most~$O(d/r)$ stages, allowing us to take $\eps = (d/r) \tilde{\eps}$ rather than $\tilde{\eps} \ell_{\max}$ (a significant improvement when $r = \Theta(d)$).

    The idea is now for~$\calA'$ to choose a slightly different~$\bR_t$ in each round, of cardinality~$\br_t \geq 1$, so that $(\bsigma_t)_{ii} \geq \Omega(\frac{\tr \bsigma_t}{\br_t \ln r}) $. 
    (Note that we need not be concerned with the sum of $(\bsigma_t)_{ii}$ on the new~$\bR_t$, since we're now only using that $|\bR_t| \geq 1$ always.)
    If we can show this is possible, then we can use it as a replacement for \Cref{ineq:fuz} when deriving \Cref{ineq:33}; we'll then get    
    \begin{equation}    
        \DBchiminus{\bL}{\rho}{\bsigma'} \leq  \sum_{t < \bell} \frac{O(\br_t \ln r)}{\tau_t} \cdot \tau_t f/m_\delta \leq O( d (\ln r) f/m_\delta) = O(\eps),
    \end{equation}
    as claimed.

    But the proof that we can choose~$\bR_t$ as described is elementary.  Essentially, the algorithm has a nonincreasing sequence of (at most) $r$~numbers $x_1, \dots, x_r$ (where $x_i = (\bsigma_t)_{\bd_t+1-i, \bd_t+1-i}$) whose sum is (at least)~$\bs$. 
    We need to show that for some~$\br_t$ it holds that $x_{\br_t} \geq \Omega(\frac{\bs}{\br_t \ln r})$. But if $x_{k} \ll \frac{\bs}{k \ln r}$ for all $k \in [r]$, then $\sum_{k=1}^r x_k \ll \bs$, a contradiction.
\end{proof}

\subsection{Conclusions from the central estimation algorithm} \label{sec:conclude}
\begin{corollary}
    After applying \Cref{thm:central} and introducing the quantum state $\wh{\brho} = \wt{\brho}_{\mid \bR} = \frac{1}{1-\beps'} \wt{\brho}[\bR]$  (extended with~$0$'s so it is in $\C^{d \times d}$), we have
    \begin{equation}    \label[ineq]{ineq:corcor}
        \DBsq{\rho}{\wh{\brho}} \leq O(\eps).
    \end{equation}
\end{corollary}
\begin{proof}
    Let us write $\rho_{\mid \bR} = \frac{1}{1-\btau}\brho[\bR]$ (which again we'll extend to $\C^{d \times d}$ when necessary).
    Let us first show
    \begin{equation}    \label[ineq]{ineq:coco}
        \DBchi{\rho_{\mid \bR}}{\wt{\brho}_{\mid \bR}} \leq O(\eps).
    \end{equation}
    Up to a slight ``rescaling'' by the factors $1 - \btau$ and $1 - \beps'$, this is nearly the same as the \Cref{enum:C} conclusion,  
    \begin{equation} \label[ineq]{ineq:yup}
        \DBchiminus{\bL}{\rho}{\wt{\brho}} \leq O(\eps).
    \end{equation}
    \Cref{enum:A} tells us that $\btau, \beps' \leq O(\tilde{\eps})$ (which may be assumed at most, say,~$\frac12$); from this,  it is not hard to show that the ``rescaling''  only makes a constant-factor difference to the off-diagonal $\chi^2$-divergence contributions.  
    So to establish \Cref{ineq:coco}, it remains to analyze the effect of rescaling on the on-diagonal $\chi^2$-divergence contributions.
    Writing $\rho_{ii} = (1 + \bzeta_i) \bq_i$ for some numbers $\bzeta_i > 0$, the bound on just the diagonal contribution in \Cref{ineq:yup} is equivalent to
    \begin{equation}    \label[ineq]{ineq:budd}
        \sum_{i \in \bR} \bzeta_i^2 \bq_i \leq O(\eps).
    \end{equation}
    Then in the rescaling, when $\rho_{ii}$ is replaced by $\frac{1}{1-\btau} \rho_{ii}$ in $\rho_{\mid \bR}$ and $\bq_i$ is replaced by $\frac{1}{1-\beps'} \bq_i$ in~$\wh{\rho}$, it is as though $\bzeta_i$ is replaced by $\frac{1-\beps'}{1-\btau} \bzeta_i = (1 \pm O(\tilde{\beps})) \bzeta_i$.
    Putting that into  \Cref{ineq:budd} shows that the rescaling only changes the on-diagonal $\chi^2$-divergence contribution by an additive $O(\tilde{\eps})\sum_{i \in \bR} \bq_i = O(\tilde{\eps})$, sufficient to complete the proof of \Cref{ineq:coco}.

    Having established \Cref{ineq:coco} (and recalling $\wh{\brho} = \wt{\brho}_{\mid \bR}$), \Cref{thm:quantum-rev-ineq}  immediately implies 
    \begin{equation} \label[ineq]{ineq:who}
        \DBsq{\rho_{\mid \bR}}{\wh{\brho}} \leq O(\eps).
    \end{equation}
    On the other hand, the Gentle Measurement Lemma~\cite{Win99}
    \begin{equation}
        \tr \rho[\bR] = 
        \Fid{\rho}{\rho_{\mid \bR}}^2 = (1 -\tfrac12 \DBsq{\rho}{\rho_{\mid\bR}})^2.
    \end{equation}    
     From $\tr \rho[\bR] = 1 - \btau  \geq  1-O(\tilde{\eps})$, the above directly yields $\DBsq{\rho}{\rho_{\mid\bR}} \leq O(\tilde{\eps})$, and thus \Cref{ineq:corcor} follows from \Cref{ineq:who} and $\DB{{\cdot}}{{\cdot}}$ being a metric.
\end{proof}
Now by working out the parameters (using both \Cref{thm:central,thm:central2}), we get the below  Frobenius-to-infidelity transformation.  The further transformation to relative entropy accuracy promised in \Cref{thm:main} follows by applying \Cref{thm:learn-KL-easy}.
\begin{corollary}   \label{cor:log}
    A state estimation algorithm with Frobenius-squared rate $f = f(d,r) \gg \log d$ may be transformed (preserving the single-copy measurement property) into a state estimation algorithm with the following property:

        Given parameters $\eps, r$, and $M$ copies of a quantum state $\rho \in \C^{d \times d}$ of rank at most~$r$, either 
        \begin{equation}
            M = O\parens*{\frac{r f(d,r)}{\eps} \cdot \log^2(1/\eps) \log\log(1/\eps)} \quad \text{or alternatively,} \quad M = O\parens*{\frac{1}{\eps} \cdot d^2 f(d,r)(\log d)(\log r)},
        \end{equation}
        suffices for  the algorithm to output (with probability at least~$.99$) the classical description of a quantum state~$\wh{\brho}$ with infidelity $1 - \Fid{\rho}{\wh{\brho}} =\frac12 \DBsq{\rho}{\wh{\brho}} \leq \eps$.

        In particular, by \Cref{thm:ow}
        \begin{equation}
            M = O\parens*{\frac{r d}{\eps} \cdot \log^2(1/\eps) \log\log(1/\eps)} \text{ suffices using collective measurements}
        \end{equation}
        (or for very small~$d$, alternatively $M = O\parens*{\frac{1}{\eps} \cdot d^3 (\log d)(\log r)}$ suffices).
        And,  by \Cref{thm:krt}
        \begin{equation}
            M = O\parens*{\frac{r^2 d}{\eps} \cdot \log^2(1/\eps) \log\log(1/\eps)} \text{ suffices using single-copy measurements}        
        \end{equation}
        (or for very small~$d$, alternatively $M = O\parens*{\frac{1}{\eps} \cdot rd^3 (\log d)(\log r)}$).
\end{corollary}

It remains to obtain the Frobenius-to-$\chi^2$ transformation promised in \Cref{thm:main}.  This is \Cref{cor:yoyo} below, which we achieve in two steps.
\begin{corollary}   \label{cor:cookie}
    After applying \Cref{thm:central} and introducing the quantum state $\wh{\brho} = \eta \cdot \frac{1}{\bL} \Id_{\bL \times \bL} + (1-\eta) \wt{\brho}$ (where we assume $\eta < 1/2$, say)  we have 
    \begin{align}
        \DBoff{\rho}{\wh{\brho}} &\leq O(\tilde{\eps}\log(1/\tilde{\eps}) + (d/r)(\tilde{\eps}^2/\eta)) \\
        \DBon{\rho}{\wh{\brho}} &\leq O(\eta + \tilde{\eps}\log(1/\tilde{\eps}) + (d/r)(\tilde{\eps}^2/\eta)) \\
        \implies \quad 
        \DBchi{\rho}{\wh{\brho}} &\leq O(\eta + \tilde{\eps}\log(1/\tilde{\eps}) + (d/r)(\tilde{\eps}^2/\eta)), \label[ineq]{ineq:love}
    \end{align}
    where we have split out the ``off-diagonal'' and ``on-diagonal'' contributions to~$\DBchi{\rho}{\wh{\brho}}$.  (Also, the factors of ``$d$'' in the above three bounds may be replaced by~$|\bL|$, which is potentially much smaller.)
\end{corollary}
\begin{remark}  \label{rem:cookies}
    Given this corollary, it would be natural to fix
    \begin{equation}
    \eta = \sqrt{d/r} \cdot \tilde{\eps}
    \end{equation} so as to balance the $\eta$~term and the $(d/r)(\tilde{\eps}^2/\eta)$~term,
    making both contribute $O(\sqrt{d/r} \cdot \tilde{\eps})$.
    This swamps the~$\tilde{\eps}\log(1/\tilde{\eps})$ term in \Cref{ineq:love} (up to a log factor).
    Thus with this choice of~$\eta$ we get the bound $\DBchi{\rho}{\wh{\brho}} \leq \wt{O}(\sqrt{d/r} \cdot \tilde{\eps})$, and if we want this to equal some ``$\eps_{\text{final}}$'' then we need to choose 
    \begin{equation}
        \tilde{\eps} = \wt{\Theta}(\sqrt{r/d} \cdot \eps_{\text{final}}),
    \end{equation}
    thereby making the final copy complexity $\wt{O}(rf/\tilde{\eps}) = \wt{O}(\sqrt{rd} \cdot f / \eps_{\text{final}})$, as stated in \Cref{thm:main}.
    Note that with these choices, the smallest eigenvalue of~$\wh{\brho}$ will be~$\wt{\Omega}(\eps_{\text{final}}/d)$, as expected.

    The reason we do not simply directly fix $\eta = \sqrt{d/r} \cdot \tilde{\eps}$ in the proof of \Cref{cor:cookie} is that in our later application to quantum zero mutual information testing (\Cref{subsec:testingzero}), it will be important to allow for a tradeoff between the off-diagonal $\chi^2$-error, the on-diagonal $\chi^2$-error, and the minimum eigenvalue of~$\wh{\brho}$. 
\end{remark}
\begin{proof}[Proof of \Cref{cor:cookie}.]
    Recall that in \Cref{thm:central}, we have 
    \begin{equation}
        \eps = \tilde{\eps}\log(1/\tilde{\eps});
    \end{equation}
    we will use this shorthand throughout the present proof.
    We start by employing  \Cref{ineq:dbl}:
    \begin{equation}    
        \DBchi{\rho}{\wh{\brho}} \leq \DBchi{\rho[\bL]}{\wh{\brho}[\bL]}  +  \DBchiminus{\bL}{\rho}{\wh{\brho}} = \DBoff{\rho[\bL]}{\wh{\brho}[\bL]} + \DBon{\rho[\bL]}{\wh{\brho}[\bL]}  +  \DBchiminus{\bL}{\rho}{\wh{\brho}}.
    \end{equation}
    To bound the off-diagonal contribution, we use: (i)~each diagonal entry of $\wh{\brho}[\bL]$ is at least $\eta/\bL$; (ii)~the off-diagonal Frobenius-squared of $\rho - \wh{\brho}$ is the same as that of $\rho - \wt{\brho}$  (since $\wh{\brho}, \wt{\brho}$ are both diagonal), which is bounded by $O(\tilde{\eps}^2)/(r \ln(1/\delta))$ from \Cref{enum:B}.  Combining these facts yields
    \begin{equation} \label[ineq]{ineq:light1}
        \DBoff{\rho[\bL]}{\wh{\brho}[\bL]}  \leq O((\abs{\bL}/r)(\tilde{\eps}^2/\eta)).
    \end{equation}
    As for the on-diagonal contribution, writing $p$ for the diagonal entries of~$\rho[\bL]$, and $\wh{\bq} = (\eta/\bL)\Id_{\bL} + (1-\eta) \bq$ for those of $\wh{\brho}$,
    \begin{equation}
        \DBon{\rho[\bL]}{\wh{\brho}[\bL]}
        = \sum_{i \in \bL} \frac{(p_i - \wh{\bq}_i)^2}{\wh{\bq}_i}
        \leq O(1) \sum_{i \in \bL} \frac{(p_i - \bq_i)^2}{\wh{\bq}_i} + O(1) \sum_{i \in \bL} \frac{(\eta/\bL)^2}{\wh{\bq}_i} + O(1) \sum_{i \in \bL} \frac{eta^2 \bq_i^2}{\wh{\bq}_i}.
    \end{equation}
    For the first two summands above we use $\wh{\bq}_i \geq \eta^2/\abs{\bL}$ in the denominator; for the third,  $\wh{\bq}_i \geq (1-\eta) \bq_i \geq \frac12 \bq_i$.  Thus
    \begin{equation}    \label[ineq]{ineq:light2}
        \DBon{\rho[\bL]}{\wh{\brho}[\bL]}  \leq O(|\bL|/\eta) \|p[\bL] - \bq[\bL]\|_2^2 + O(\eta) + O(\eta^2) \sum_{i \in \bL} \wh{\bq}_i \leq O((\abs{\bL}/r)(\tilde{\eta}^2/\eta) + \eta),
    \end{equation}
    where we used $\|p[\bL] - \bq[\bL]\|_2^2 \leq \|\rho[\bL] - \wt{\brho}[\bL]\|_\Fro^2$ and then \Cref{enum:B} again.  
    
    In light of \Cref{ineq:light1,ineq:light2}, it suffices  to show  that $\DBchiminus{\bL}{\rho}{\wh{\brho}} = \DBchiminus{\bL}{\rho}{(1-\eta)\wt{\brho}} \leq O(\eta + \eps)$, with the off-diagonal contribution being just~$O(\eps)$.
    This off-diagonal contribution differs from that of \mbox{$\DBchiminus{\bL}{\rho}{\wt{\brho}}$} by a factor of at most $1/(1-\eta) = O(1)$, so we may indeed bound it by $O(\eps)$ from \Cref{enum:C}.
    Finally, the on-diagonal contribution to $\DBchiminus{\bL}{\rho}{\wh{\brho}}$ is
    \begin{equation}
        \sum_{i \not \in \bL} \frac{(\rho_{ii} - (1-\eta)\bq_i)^2}{(1-\eta) \bq_i} \leq O(1) \sum_{i \not \in \bL} \frac{(\rho_{ii} - \bq_i)^2}{\bq_i} + O(1) \sum_{i \not \in \bL}\eta^2 \bq_i \leq O(\DBchiminus{\bL}{\rho}{\wt{\brho}}) + O(\eta^2),
    \end{equation}
    and this is $O(\eta^2 + \eps) \leq O(\eta + \eps)$, as required.
\end{proof}

Working out the parameters (just using \Cref{thm:central}), along the lines of \Cref{rem:cookies}, we get:
\begin{corollary}   \label{cor:yoyo}
    A state estimation algorithm with Frobenius-squared rate $f = f(d,r) \gg \log d$ may be transformed (preserving the single-copy measurement property) into a state estimation algorithm with the following property:

        Given parameters $\eps, r$ (with $\eps \leq 1/2$), and $M$ copies of a quantum state $\rho \in \C^{d \times d}$ of rank at most~$r$, 
        \begin{equation}
            M = \wt{O}\parens*{\frac{\sqrt{rd} \cdot f(d,r)}{\eps}}
        \end{equation}
        suffices\footnote{The hidden $\polylog$ terms are at most $\log^2(d/\eps) \log \log(d/\eps)$, but may be optimized further in special cases~\cite{Tang22}. 
        In case $r = d$, one may take $M = O(\frac{d f(d)}{\eps} \log^2(1/\eps) \log \log(1/\eps))$, so the $\polylog$ terms have no dependence on~$d$. In case $r = O(1)$, if $\eps \geq \exp(-\Omega(\sqrt{d}))$, one may take $M = O(\frac{\sqrt{d} f(d)}{\eps} \log(d/\eps) \log \log(d/\eps))$.}
        for the algorithm to output (with probability at least~$.99$) the classical description of a quantum state~$\wh{\brho}$ with $\DBchi{\rho}{\wh{\brho}} \leq \eps$.
        
        In particular, by \Cref{thm:ow}
        \begin{equation}
            M = \wt{O}\parens*{\frac{r^{.5} d^{1.5}}{\eps}} \text{ suffices using collective measurements}.
        \end{equation}
        And,  by \Cref{thm:krt}
        \begin{equation}
            M = \wt{O}\parens*{\frac{r^{1.5} d^{1.5}}{\eps}} \text{ suffices using single-copy measurements.}
        \end{equation}
\end{corollary}

\subsection{A simple Frobenius-squared estimator} \label{sec:simple-frob}
\begin{proposition}     \label{prop:simple}
    There is estimation algorithm for quantum states, making single-copy measurements, achieving expected Frobenius-squared error~$O(d^2/m)$.
\end{proposition}
\begin{proof}
    It suffices to achieve Frobenius-squared error $d/m$ using $O(dm)$ copies.
    Assume for simplicity $d$ is even. (We leave the odd $d$ case to the reader.)
    Then there is a simple way~\cite{Wik} to construct a partition~$\calP$ of the edges of the complete graph on $d$~vertices into $d-1$~matchings.
    Fix a particular matching $M \in \calP$, and associate to it the POVM with elements
    \begin{equation}    \label{eqn:povmmm}
        (X^{\pm}_{ij})_{\{i,j\} \in M}, \quad\text{where}\quad
        X^{\pm}_{ij} = \frac12(\ket{i}\!\bra{i} + \ket{j}\!\bra{j}) \pm \frac12(\ket{i}\!\bra{j} + \ket{j}\!\bra{i}).
    \end{equation}
    When we measure~$\rho$ with this POVM, we obtain outcome $(\{i,j\},\pm)$ with probability
    \begin{equation}    \label{eqn:reallyreal}
        p_{ij} \pm r_{ij}, \quad p_{ij} \coloneqq \avg\{\rho_{ii},\rho_{jj}\}, \quad r_{ij} \coloneqq \Re \rho_{ij}.
    \end{equation}
    If we similarly define a POVM $(Y_{ij}^{\pm})$ but with a factor of~$\mathrm{i} = \sqrt{-1}$ in the off-diagonal elements of \Cref{eqn:povmmm}, we will similarly get outcomes with probabilities $p_{ij} \pm s_{ij}$, where $s_{ij} \coloneqq \Im \rho_{ij}$.
    We focus on analyzing \Cref{eqn:reallyreal}, as the imaginary-part analysis will be identical.

    Suppose we now measure this POVM~$m$ times and form the random variables $\wh{\br}_{ij}$ (for $\{i,j\} \in M$), where 
    \begin{equation}
        \wh{\br}_{ij} = \frac{\boldf^+_{ij} - \boldf^-_{ij}}{2}, \quad\text{with}\quad  \boldf^\pm_{ij} = \text{fraction of outcomes that are $(\{i,j\},\pm)$}.
    \end{equation}
    Then $\E[\wh{\br}_{ij}] = r_{ij}$, and 
    \begin{equation}
        \E[(r_{ij} - \wh{\br}_{ij})^2]
        = \Var[\wh{\br}_{ij}] \leq 
        \frac12 \Var[\boldf^+_{ij}] + \frac12 \Var[\boldf^-_{ij}] 
        \leq \frac{p_{ij} + r_{ij}}{2m} + \frac{p_{ij} - r_{ij}}{2m} = p_{ij}/m.
    \end{equation}
    Repeating the analysis for the imaginary parts, we use $2m$ copies of~$\rho$ to get estimates for all $\rho_{ij}$, $\{i,j\} \in M$, achieving total expected squared-error 
    \begin{equation}
        \sum_{\{i,j\} \in M} 2p_{ij}/m = \sum_{i \in [d]} \rho_{ii}/m = 1/m.
    \end{equation}
    Repeating this for all $M \in \calP$ uses $O(dm)$ copies of~$\rho$ and gives estimates for all off-diagonal elements of~$\rho$, with total expected squared-error $(d-1)/m$.
    Finally, we can use standard basis measurements to estimate the diagonal elements of~$\rho$, using \Cref{prop:learn-L2}: $m$ more copies of~$\rho$ suffice to achieve total expected squared-error~$1/m$.
\end{proof}

\section{Testing zero mutual information}   \label{sec:testing}

We now move on to showing the main application of our $\chi^2$ tomography algorithm: testing zero quantum mutual information. 
We will explain below in \Cref{subsec:testingzero} how our $\chi^2$ tomography algorithm is crucial to achieving this result. 
But first, we introduce and analyze a variant of the quantum mutual information that features in our analysis.

\subsection{Mutual information versus its Hellinger variant}

The goal of this subsection is to prove the below theorem, showing that the standard quantum mutual information is not much larger than the ``Hellinger mutual information'':
\begin{theorem}\label{thm:Hell-mutual-info}
    Let $\rho = \rho_{AB}$ be a bipartite quantum state on $A \otimes B$, where $A \cong B \cong \C^{d}$.
    Writing $\DHellsq{\rho}{\rho_A \otimes \rho_B} = \eta$, it holds that $I(A : B)_\rho \leq \eta \cdot O(\log(d/\eta))$.
\end{theorem}
We also observe that by restricting~$\rho_{AB}$ to be diagonal, we immediately obtain the analogous theorem concerning classical mutual information.
We remark that proving this classical version directly is no easier than proving the quantum version.
\begin{corollary}\label{cor:hell-mutual-info}
    Let $p = p_{AB}$ be bipartite classical state on $A \times B$, where $|A| = |B| = d$.
    Writing \mbox{$\dhellsq{p}{p_A \times p_B} = \eta$}, it holds that $I(A : B)_p \leq \eta \cdot O(\log(d/\eta))$.
\end{corollary}

We first state a bound on the continuity of mutual information in terms of the trace distance and the subsystem dimension. 
A bound of the following form can be proven a number of ways, for example by appealing to the Petz--Fannes--Audenaert~\cite{Audenaert2007}  and Alicki--Fannes~\cite{Alicki2004} inequalities; see~\cite[Appendix F]{Flammia2017}. 
The bound we use is an immediate corollary of \cite[Prop.\ 1]{Shirokov2017} which gives small explicit constants. 

\begin{lemma}[Continuity of quantum mutual information] 
\label{lem:contmutualinfo}
    For two density operator $\rho$, $\sigma$ on $A\otimes B$ with $A\cong B \cong \mathbb{C}^d$ and $\tfrac{1}{2}\|\rho -\sigma\|_1 = \eps$, we have
    \begin{equation}
        |I(A:B)_\rho - I(A:B)_\sigma| \le 2 \eps \log\frac{4d}{\eps}\,.
    \end{equation}
\end{lemma}

\begin{proof}
    The bound from \cite[Prop.\ 1]{Shirokov2017} for the quantum conditional mutual information applies immediately with a trivial conditioning system and with the bound $\eps \le 1$ to get the constant 4. 
\end{proof}

We also must bound how much $\DHellsq{\rho}{\rho_A\otimes\rho_B}$ changes relative to a depolarization of $\rho$. 

\begin{lemma} \label{lem:DHellsqdepolbound}
Given a density operator $\rho$ on $A\otimes B$ with $A\cong B\cong\mathbb{C}^d$ and the depolarization $\sigma = \Delta_\eps(\sigma) = (1-\eps)\rho + \frac{\eps}{d^2}\mathbbm{1}$, the squared Hellinger distance obeys
\begin{equation}
    \DHellsq{\sigma}{\sigma_A\otimes\sigma_B} \le C \sqrt{\eps} + \DHellsq{\rho}{\rho_A\otimes\rho_B}\,,
\end{equation}
where the constant can be chosen as $C = 4+4\sqrt{2}$.
\end{lemma}

\begin{proof}
    Since $\DHell{\rho}{\sigma}$ is a metric, we have
    \begin{equation}
        \DHell{\sigma}{\sigma_A\otimes\sigma_B} \le \DHell{\sigma}{\rho} + \DHell{\rho}{\rho_A\otimes\rho_B} + \DHell{\rho_A\otimes\rho_B}{\sigma_A\otimes\sigma_B}\,.
    \end{equation}
    By \Cref{prop:quantum-ineqs} we have $\DHell{\rho}{\sigma}\le \sqrt{2\Dtr{\rho}{\sigma}}$ and we have $\Dtr{\rho}{\sigma} \le \eps$ by the triangle inequality. 
    By the triangle inequality and quantum data processing inequality (monotonicity of trace distance under CPTP maps), we also have $\Dtr{\rho_A\otimes\rho_B}{\sigma_A\otimes\sigma_B} \le \Dtr{\rho_A}{\sigma_A} + \Dtr{\rho_B}{\sigma_B} \le 2\eps$. 
    Plugging this in above and rearranging, we have
    \begin{equation}
        \DHell{\sigma}{\sigma_A\otimes\sigma_B} - \DHell{\rho}{\rho_A\otimes\rho_B} \le (2+\sqrt{2})\sqrt{\eps}\,.
    \end{equation}
    For any two states $\rho, \sigma$, $\DHell{\rho}{\sigma} \le \sqrt{2}$, so we can multiply both sides by $\DHell{\sigma}{\sigma_A\otimes\sigma_B} + \DHell{\rho}{\rho_A\otimes\rho_B} \le 2\sqrt{2}$ and the bound follows.
\end{proof}

Now we can prove~\Cref{thm:Hell-mutual-info}. 

\begin{proof}[Proof of \Cref{thm:Hell-mutual-info}.]
We begin by smoothing the state $\rho$ with a depolarizing channel to obtain 
\begin{equation}
    \sigma = (1-\eps) \rho + \frac{\eps}{d^2}\mathbbm{1}\,.
\end{equation}
By the triangle inequality, we have $\tfrac{1}{2}\|\rho - \sigma\|_1 = \tfrac{\eps}{2} \| \rho - \mathbbm{1}/d^2\|_1 \le \eps$. 
The change in the mutual information by passing from $\rho \to \sigma$ is bounded using \Cref{lem:contmutualinfo}, 
\begin{equation}
    I(A:B)_\rho \le I(A:B)_\sigma + 2 \eps \log\frac{4d}{\eps}\,.
\end{equation}
Note that $I(A:B)_\rho \ge I(A:B)_\sigma$ by the quantum data processing inequality for relative entropy (monotonicity of mutual information under local CPTP maps). 

By \Cref{thm:quantum-rev-ineq}, we have 
\begin{equation}
    I(A:B)_\sigma \le (2+ \Dren{\infty}{\sigma}{\sigma_A\otimes\sigma_B}) \cdot \DHellsq{\sigma}{\sigma_A\otimes\sigma_B}\,.
\end{equation}
Since $\sigma - \frac{\eps}{d^2}\mathbbm{1} \ge 0$, the positivity of the partial trace map shows that the reduced states $\sigma_A$ and $\sigma_B$ satisfy
\begin{equation}
    \sigma_A \ge \frac{\eps}{d}\mathbbm{1} \,, \text{ and}\quad \sigma_B \ge \frac{\eps}{d}\mathbbm{1}\,.
\end{equation}
Therefore the R\'{e}nyi entropy term is bounded using \Cref{fact:Dren2} as 
\begin{equation}
    \Dren{\infty}{\sigma}{\sigma_A\otimes\sigma_B}) \le \log \|\sigma_A^{-1}\otimes\sigma_B^{-1}\| \le \log \frac{d^2}{\eps^2}\,.
\end{equation}
Using \Cref{lem:DHellsqdepolbound}, the $\DHellsq{\sigma}{\sigma_A\otimes\sigma_B}$ term is bounded by
\begin{equation}
    \DHellsq{\sigma}{\sigma_A\otimes\sigma_B} \le C \sqrt{\eps} + \DHellsq{\rho}{\rho_A\otimes\rho_B} = C \sqrt{\eps} + \eta\,.
\end{equation}
Putting this all together, we have that 
\begin{equation}
    I(A:B)_\rho 
    \le \Bigl(2+ \log \frac{d^2}{\eps^2}\Bigr) \cdot \bigl(C \sqrt{\eps} + \eta\bigr) + 2 \eps \log\frac{4d}{\eps} 
    \le (\eta + \sqrt{\eps}) \, O\bigl(\log (d/\eps)\bigr)\,.
\end{equation}
Choosing $\eps = O(\eta^2)$ completes the proof.
\end{proof}

\subsection{Testing zero classical mutual information}
Before moving to the trickier quantum case, we warm up by showing 
an efficient tester for classical mutual information, establishing \Cref{thm:test-classical}. 
First we give a short proof of the following (which appears explicitly as~\cite[Lem.~7]{Acharya2015}):
\begin{lemma}   \label{lem:learn-marginals}
    Given $n = O(d/\eps)$ samples from two distributions $q, q'$ on~$[d]$, one can output a hypothesis $\wh{\bq} \times \wh{\bq}'$ that (with probability at least~$.99$) satisfies $\dchisq{q \times q'}{\wh{\bq} \times \wh{\bq}'} \leq \eps$.
\end{lemma}
\begin{proof}
    It suffices to separately learn each of $q,q'$ to $\chi^2$-accuracy~$\eps/3$ and high confidence (which can be done applying \Cref{prop:learn-chi2} and Markov's inequality), and then apply the below \Cref{prop:chi-add}.
\end{proof}
\begin{proposition} \label{prop:chi-add}
    Let $\dchisq{p}{q} = \eps_1$, $\dchisq{p'}{q'} = \eps_2$. 
    Then we have the near-additivity formula
    \begin{equation}
        \dchisq{p\otimes p'}{q \otimes q'} = (1+\eps_1)(1+\eps_2) - 1 = \eps_1 + \eps_2 + \eps_1 \eps_2.
    \end{equation}
\end{proposition}
\begin{proof}
    This follows essentially immediate from the second formula in \Cref{def:chi2}.
\end{proof}

Now to prove \Cref{thm:test-classical}, suppose we are given access to a bipartite probability distribution~$p = p_{AB}$ on $[d] \times [d]$ and we are promised that either $I(A : B)_p = 0$ or $I(A : B)_p \geq \eps$.
In the former case we have
\begin{equation} \label{ineq:c}
    I(A : B)_p = 0 \quad\implies\quad p = p_A \times p_B \quad \implies \quad \dchisq{p}{\wh{\bp}_A \times \wh{\bp}_B} \leq c \eps/\log(d/\eps)
\end{equation}
with high probability if we estimate the two marginals using $O((d/\eps)\cdot \log(d/\eps))$ samples as in \Cref{lem:learn-marginals}. 
(Here $c > 0$ may be any small constant.)
On the other hand, in case $I(A : B)_p \geq \eps$, we get $\dhellsq{p}{p_A \times p_B} \geq \Omega(\eps/\log(d/\eps))$ from \Cref{cor:hell-mutual-info}; ensuring that the constant~$c$ in \Cref{ineq:c} is small enough, this implies
\begin{equation}
    \dhellsq{p}{\wh{\bp}_A \times \wh{\bp}_B} \geq \Omega(\eps/\log(d/\eps))
\end{equation}
also.
Now again ensuring~$c$ is small enough, our classical mutual information tester \Cref{thm:test-classical} follows by using the below ``$\chi^2$-vs.-$\mathrm{H}^2$ identity tester'' of Daskalakis--Kamath--Wright with the ``known'' distribution being $\wh{\bp}_A \times \wh{\bp}_B$.  This distinguishes our  two cases (with high probability) using $O(\sqrt{d\times d}/\eps')$ samples, $\eps' =\Theta(\eps/\log(d/\eps))$; in other words, with $O((d/\eps) \log(d/\eps))$ samples.

\begin{theorem} \label{thm:dkw} (\cite[Thm.~1]{DKW2018}.)
    For any ``known'' distribution $q$ on~$[d]$, there is a testing algorithm with the following guarantee:
    Given $0 < \eps \leq 1/2$, as well as $n = O(\sqrt{d}/\eps)$ samples from an unknown distribution $p$~on~$[d]$, if $\dchisq{p}{q} \leq \eps/2$ then the test accepts with probability at least~$.99$, and if $\dhellsq{p}{q} \geq \eps$ then the test rejects with probability at least~$.99$.
\end{theorem}

\subsection{Testing zero quantum mutual information}\label{subsec:testingzero}
To prove \Cref{thm:test}, we now endeavor to repeat the result from the previous setting in the quantum case. 
Naturally, it is crucially important that we are able to do quantum state tomography with respect to Bures $\chi^2$-divergence.
This lets use the following quantum 
``$\chi^2$-vs.-$\mathrm{H}^2$ identity tester'' from~\cite{Badescu2019} in place of \Cref{thm:dkw}.
\begin{theorem} \label{thm:bow} (\cite[Thm.~1]{Badescu2019}.)
    For any ``known'' quantum state $\sigma \in \C^{d \times d}$, there is a testing algorithm with the following guarantee:
    Given $0 < \eps \leq 1/2$, as well as $n = O(d/\eps)$ copies of an unknown state $\rho \in \C^{d \times d}$, if $\DBchi{\rho}{\sigma} \leq .49\eps$ then the test accepts with  probability at least~$.99$, and if $\DHellsq{\rho}{\sigma} \geq 2\eps$ then the test rejects with probability at least~$.99$.
\end{theorem}
We can relate quantum mutual information and its Hellinger version using \Cref{thm:Hell-mutual-info} in place of \Cref{cor:hell-mutual-info}.
It would seem then that we could establish our quantum zero mutual information tester \Cref{thm:test} in exactly the same way we did its classical analogue, using $\wt{O}(d^2/\eps)$ copies of a bipartite $d \times d$-dimensional state~$\rho$.
Unfortunately, there is a missing piece: a quantum analogue of \Cref{lem:learn-marginals} with $O(d^2/\eps)$ copy complexity.
We prove a slightly worse variant, which leads to our main testing \Cref{thm:test}:
\begin{theorem}   \label{thm:suffering}
    There is a tomography algorithm that, given parameter $0 < \eps \leq 1/2$ and 
    \begin{equation}    \label{eqn:suff}
        n = \max\{\wt{O}(r d^{1.5}/\eps),\;\wt{O}(r^{.5} d^{1.75}/\eps)\}
    \end{equation} 
    copies of unknown $d$-dimensional quantum states $\xi, \rho$ of rank at most~$r$, outputs  diagonal states $\bsigma, \btau$ such that (with probability at least~$.9$),
    \begin{equation}
        \DBchi{\xi \otimes \rho}{\bsigma' \otimes \btau'} \leq \eps.
    \end{equation}
\end{theorem}
\begin{proof}
    The strategy is to apply our central estimation algorithm in the form of \Cref{cor:cookie} to both~$\xi, \rho$ (with the Frobenius-learner from \Cref{thm:ow} with $f(d,r) = O(d)$).
    We use the parameter choices 
    \begin{equation}
        \eta = \eps, \qquad \tilde{\eps} = \eps \cdot \min\{1/d^{.5}, r^{.5}/d^{.75}\}.
    \end{equation}
    This leads to the claimed copy complexity from \Cref{eqn:suff}, and yields (with probability at least~$.9$) estimates~$\bsigma, \btau$ satisfying
    \begin{align}
        \DBoff{\xi}{\bsigma},\  \DBoff{\rho}{\btau} &\leq \wt{O}(\eps \cdot (1/d^{.5} + \min\{1/r, r^{.5}/d^{.75}\})) = \wt{O}(\eps/\sqrt{d}), \label[ineq]{ineq:z1}\\
        \DBon{\xi}{\bsigma},\  \DBoff{\rho}{\btau} &\leq \wt{O}(\eps),  \label[ineq]{ineq:z2}
    \end{align}
    with $\bsigma, \btau$ having minimum eigenvalue at least~$\eps/d$.
    We claim that for any fixed outcomes $\bsigma = \sigma'$ and $\btau' = \tau'$ satisfying the above, it holds that
    \begin{equation}    \label[ineq]{eqn:yok}
        \DBchi{\xi \otimes \rho}{\sigma' \otimes \tau'} \leq \wt{O}(\eps).
    \end{equation}
    This is sufficient to complete the proof.    \\
    
    To establish \Cref{eqn:yok}, let us write $\sigma' = \diag(s')$ and $\tau' = \diag(t')$, with
    \begin{equation}    \label[ineq]{ineq:co}
        s'_a \geq \eps/d, \quad t'_j \geq \eps/d \qquad \text{for all $a,j  \in [d]$.}
    \end{equation}
    We will break up $\DBchi{\xi \otimes \rho}{\sigma' \otimes \tau'}$ into three parts: on-on-diagonal, on-off-diagonal, and off-off-diagonal:
    \begin{equation}
        \underbrace{\sum_{a,i = 1}^d \tfrac{1}{s'_at'_i} (\xi_{aa}\rho_{ii} - s'_at'_i)^2}_{\textrm{(ON-ON)}}
        + 
        \underbrace{\parens*{
        \sum_{a = 1}^d \sum_{i \neq j} \tfrac{2}{s'_at'_i + s'_a t'_j} \abs{\xi_{aa}\rho_{ij}}^2
        +
        \sum_{a \neq b} \sum_{i=1}^d \tfrac{2}{s'_at'_i + s'_b t'_i} \abs{\xi_{ab}\rho_{ii}}^2
        }}_{\textrm{(ON-OFF)}}
        + 
        \underbrace{
        \sum_{a \neq b} \sum_{i\neq j} \tfrac{2}{s'_at'_i + s'_b t'_j} \abs{\xi_{ab}\rho_{ij}}^2}_{\textrm{(OFF-OFF)}}.
    \end{equation}
    First, using \Cref{prop:chi-add},
    \begin{equation} \label[ineq]{ineq:add}
        \textrm{(ON-ON)} = \dchisq{\diag(\xi) \otimes \diag(\rho)}{s' \otimes t'} = (1+\dchisq{\diag(\xi)}{s'})(1+\dchisq{\diag(\rho)}{t'}) - 1.
    \end{equation}
    But $\dchisq{\diag(\xi)}{s'} \leq \DBon{\xi}{\sigma'} \leq \wt{O}(\eps)$ by \Cref{ineq:z1}, and similarly for $\dchisq{\diag(\rho)}{t'}$, so we conclude from \Cref{ineq:add} that $\textrm{(ON-ON)} \leq (1 + \wt{O}(\eps))(1 + \wt{O}(\eps)) - 1 = \wt{O}(\eps)$, as needed for \Cref{eqn:yok}.

    Moving to (ON-OFF), the first term in it factorizes to
    \begin{equation}
        \label{eqn:factorize}
        \parens*{\sum_{a=1}^d \tfrac{1}{s'_a}\xi_{aa}^2}
        \parens*{\sum_{i \neq j} \tfrac{2}{t'_i + t'_j} \abs{\rho_{ij}}^2}
    \end{equation}
    The first factor above is precisely 
    \begin{equation}
        1+\dchisq{\diag(\xi)}{s'} \leq 1 + \DBon{\xi}{\sigma'} \leq 1 + \wt{O}(\eps) \leq O(1).
    \end{equation}
    The second factor in \Cref{eqn:factorize} is 
    \begin{equation}
        \label{ineq:almost}
        \sum_{i \neq j} \tfrac{2}{t'_i + t'_j} \abs{\rho_{ij}}^2 = 
        \DBoff{\rho}{\tau'} \leq \DBchi{\rho}{\tau}{t'} \leq \wt{O}(\eps).
    \end{equation}
    Thus \Cref{eqn:factorize}, and indeed both terms in (ON-OFF), can be bounded by~$\wt{O}(\eps)$, as needed for \Cref{eqn:yok}.
    
    It remains to bound (OFF-OFF) by~$\wt{O}(\eps)$.
    By the AM-GM inequality,
    \begin{equation}
        \frac{s'_at'_i+ s'_b t'_j}{2} \geq \sqrt{s'_a t'_i s'_b t'_j} = \sqrt{s'_a s'_b}\sqrt{t'_it'_j}.
    \end{equation}
    Of course there is no reverse AM-GM inequality, but we at least have
    \begin{equation}
        \sqrt{xy} \geq \min(\sqrt{x/y},\sqrt{y/x}) \cdot \frac{x+y}{2} \qquad \forall x,y > 0.
    \end{equation}
    When $(x,y)$ is $(s'_a,s'_b)$ or $(t'_i,t'_j)$, we have $\min(\sqrt{x/y},\sqrt{y/x}) \geq \sqrt{\eps/d}$ (from \Cref{ineq:co}), and hence
    \begin{equation}
          \frac{s'_at'_i+ s'_b t'_j}{2} \geq (\eps/d) \cdot \frac{s'_a + s'_b}{2} \cdot \frac{t'_i + t'_j}{2}
    \end{equation}
    Putting this into the definition of (OFF-OFF) yields 
    \begin{equation}
        \textrm{(OFF-OFF)} \leq (d/\eps) \cdot \sum_{a \neq b} \sum_{i \neq j} \tfrac{2}{s'_a + s'_b}\cdot  \tfrac{2}{t'_i + t'_j} \cdot \abs{\xi_{ab} \rho_{ij}}^2
        = (d/\eps) \parens*{\sum_{a \neq b} \tfrac{2}{s'_i + s'_j} \abs{\xi_{ab}}^2} \parens*{\sum_{i \neq j} \tfrac{2}{t'_i + t'_j} \abs{\rho_{ij}}^2}.
    \end{equation}
    The last factor here is bounded by $\DBoff{\rho}{\tau'}$ in \Cref{ineq:almost}, and the similar factor with $s'$ and $\xi$ is similarly bounded.
    Hence using \Cref{ineq:z1}, we indeed get
    \begin{equation} \label[ineq]{ineq:offoff}
        \textrm{(OFF-OFF)} \leq O(d/\eps)\cdot \wt{O}(\eps/\sqrt{d}) \cdot \wt{O}(\eps/\sqrt{d}) = \wt{O}(\eps),
    \end{equation}
    completing the proof.
\end{proof}

\section{Open Problems}\label{sec:conclusion}

One obvious and by now longstanding open question related to our work is learning in infidelity to precision~$\eps$ with $O(rd/\eps)$ samples, without any logarithms. 
This would settle the sample complexity of tomography with infidelity loss up to constant factors. 
In light of our work, perhaps we could even ask for more: 
Given our result that learning in quantum relative entropy is possible with $\wt{O}(rd/\eps)$ samples, might a similar no-logarithm bound hold here as well?

Our algorithm uses only single-copy measurements, but even these are challenging on present-day quantum computers. 
A stronger assumption on measurements is to restrict to product measurements, meaning that all POVM elements factorize into tensor products over subsystems. 
We believe this measurement model will require strictly greater sample complexity for learning in $\chi^2$-divergence and for quantum mutual information testing than the single-copy case analyzed here.

Regarding quantum mutual information testing, note that in the classical case we could learn product states to $\chi^2$-divergence well enough that the entire testing complexity was dominated by the $\chi^2$-vs.-Hellinger identity tester. 
Unfortunately, in the quantum case we couldn't quite match this. 
Might it be possible to reduce the complexity of testing zero quantum mutual information down to to $\wt{O}(d^2/\eps)$?

For learning in $\chi^2$-divergence, it would be interesting to show that $\wt{\Omega}(\sqrt{r}d^{1.5}/\eps)$ is the right lower bound; currently, we have nothing better than the infidelity-tomography lower bound of $\wt{\Omega}(rd/\eps)$. 
As explained in~\Cref{{rem:itshard}}, though, it seems like reducing the upper bound could be difficult  even for the case $r=1$. 

Although the \emph{Bures} $\chi^2$-divergence is usually the largest of the ``big four'' quantities considered in this paper, there are other quantum generalizations of $\chi^2$-divergence in the literature that are larger still than Bures $\chi^2$-divergence (see, e.g.,~\cite{Petz1996, Temme2010}).
An example is the so-called ``standard'' quantum $\chi^2$-divergence, in which the the arithmetic mean reciprocal-prefactor in \Cref{eqn:chi-formula} is replaced by a geometric mean. 
Similarly, there are also multiple generalizations of the quantum relative entropy besides the ``Umegaki'' quantum relative entropy~$\DKL{{\cdot}}{{\cdot}}$ studied herein.
As explained above, the main reason for us to consider learning with respect to \emph{Bures} $\chi^2$-divergence (as opposed to other metrics) is that it seems necessary for some applications; for example, our quantum mutual information testing problem. 
It is an interesting open question to study state tomography with respect to other generalizations of relative entropy and $\chi^2$-divergence, and in particular to decide if this is possible while still having $\wt{O}(1/\eps)$ scaling. 

More generally, a very interesting direction is to investigate for which quantum learning and testing tasks we can get away with $\wt{O}(1/\eps)$ samples, and for which we require (say) $\wt{\Omega}(1/\eps^2)$ samples.

\bibliographystyle{plainnat}
\bibliography{refs}

\end{document}